\documentclass[10pt, oneside, article]{article}
\usepackage{xcolor, dsfont}
\usepackage{graphicx}				
\usepackage{amsmath, amssymb, braket, mathrsfs}
\usepackage{algorithmicx, algpseudocode}
\usepackage{orcidlink}
\usepackage{booktabs}
\usepackage{pifont}
\usepackage{amssymb, amsmath, hep}
\usepackage{changepage}
\usepackage{enumitem}		
\usepackage[normalem]{ulem}
\usepackage{color, epstopdf}
\usepackage{geometry}  
\usepackage{amsmath}
\usepackage{amsthm}
\usepackage{amssymb}
\usepackage{algorithm}
\usepackage{color}
\usepackage[english]{babel}
\usepackage{graphicx}
\usepackage{grffile}
\usepackage{wrapfig,epsfig}
\usepackage{epstopdf}
\usepackage{url}
\usepackage{color}
\usepackage{physics}
\usepackage{epstopdf}
\usepackage{algpseudocode}
\usepackage[nameinlink,capitalize]{cleveref}
\usepackage{textcomp}
\usepackage{multirow}
\newtheorem{theorem}{Theorem}[section]
\newtheorem{lemma}[theorem]{Lemma}
\newtheorem{definition}[theorem]{Definition}

\newtheorem{proposition}[theorem]{Proposition}
\newtheorem{corollary}[theorem]{Corollary}
\newtheorem{conjecture}[theorem]{Conjecture}
\newtheorem{assumption}[theorem]{Assumption}
\newtheorem{observation}[theorem]{Observation}
\newtheorem{fact}[theorem]{Fact}
\newtheorem{remark}[theorem]{Remark}

\definecolor{SH_color}{rgb}{0.188, 0.478, 0.858}
\newcommand{\U}{\mathrm{U}}
\newcommand{\C}{\mathbb{C}}
\newcommand{\R}{\mathbb{R}}
\newcommand{\CP}{\mathbb{CP}}
\newcommand{\diag}{\mathrm{diag}}
\newcommand{\dist}{\mathrm{dist}}
\newcommand{\E}{\mathbb{E}}
\newcommand{\cmark}{\ding{51}} 

\definecolor{azure}{rgb}{0.0, 0.5, 1.0}

\DeclareMathOperator{\poly}{poly}

\definecolor{greenish}{rgb}{0.0, 0.6, 0.2}

\usepackage{soul}

\usepackage{microtype}
\usepackage{authblk}

\begin{document}
\title{Complexity and hardness of random peaked circuits}
\author[1,2]{Yuxuan Zhang\orcidlink{0000-0001-5477-8924}\thanks{quantum.zhang@utoronto.ca}} 
\affil[1]{Department of Physics and Centre for Quantum Information and Quantum Control, University of Toronto}
\affil[2]{Vector Institute for Artificial Intelligence, W1140-108 College Street, Schwartz Reisman Innovation Campus, Toronto, Ontario M5G 0C6, Canada}
\maketitle

\begin{abstract}

Near-term feasibility, classical hardness, and verifiability are the three requirements for demonstrating quantum advantage; most existing quantum advantage proposals achieve at most two. A promising candidate recently proposed is through randomly generated ``peaked circuits": quantum circuits that look random but with high output-weight on one of its output strings. 
In this work, we study an explicit construction for random peaked circuits that is closely related to the model studied in~\cite{aaronson2024verifiable}. Our construction involves first selecting a random circuit $C$ of polynomial size, which forms a $k$-design. Subsequently, a second random circuit $C'$ is chosen from the same architecture, subject to a postselection criterion: $C'$ must exhibit a high overlap with $C$ in one of their rows. The composition of these two circuits, $P = C'^\dagger C$, yields a peaked circuit where the local properties of each gate appear random. Utilizing unitary design theory properties, we demonstrate that the circuits generated by this method are non-trivial; specifically, $C'$ is provably far from $C^\dagger$. Indeed, with overwhelmingly high probability, a random peaked circuit generated this way is non-compressible and is of circuit complexity $\tilde \Omega(nk)$. This resolves an open problem posed by Aaronson in 2022~\cite{aaronson2022much}: it shows that peaked circuits selected at random are highly likely to be non-trivial. 

Secondly, employing a polynomial method, we analytically establish that estimating the peakedness of a circuit sampled from a slightly perturbed random peaked circuit distribution, to within a $2^{-\text{poly}(n)}$ additive error, is average-case \#P-hard, even when the peaked string is known. When the additive error is relaxed to $1/\text{poly}(n)$, we note that the worst-case scenario for this problem is BQP-complete. Under widely accepted assumptions on random quantum circuits, we identify a regime where no classical polynomial-time sequential simulator (that simulates quantum states gate-by-gate) attains inverse-polynomial additive accuracy on the peak on a non-negligible fraction of instances.

Thirdly, we study using peaked circuits as a practical attempt for a verifiable quantum advantage protocol. While the postselection method for generating peaked circuits could be costly, we demonstrate that numerical search for $C'$ with randomized initialization successfully returns a random peaked circuit, achieving the properties as theoretically predicted. Although numerical optimization alone cannot reach system sizes beyond the classically simulable regime, we propose a circuit stitching method that reliably generates large peaked circuits within a regime suitable for demonstrating quantum advantage. 
\end{abstract}
\newpage
\tableofcontents
\newpage
\begin{figure}
    \centering
    \includegraphics[width=.95\linewidth]{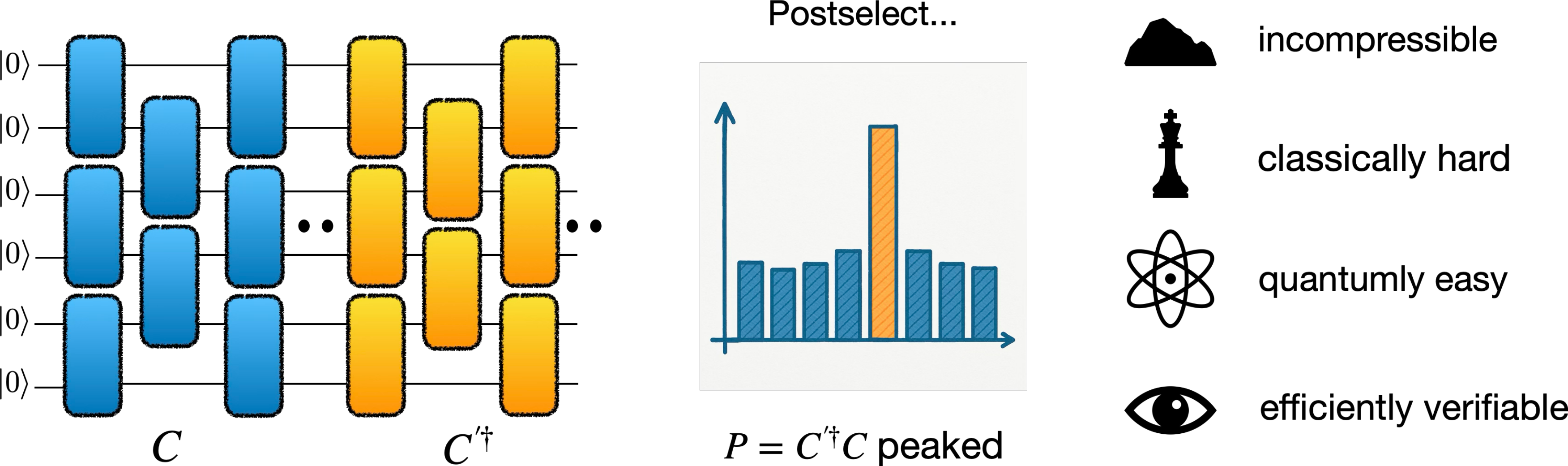}
    \caption{{\bf Overview of the results in this paper:} We consider generating peaked circuits from random circuits and postselectiing on the output distribution being peaked. We analytically prove that the peaked circuits generated this way are incompressible. Furthermore, we show that estimating the output weight of the peaked string can be average case computationally hard for a classical computer, whereas a quantum computer can obtain the peakedness by simply running the circuit and perform sampling. Compared to conventional random circuit sampling protocol, peaked circuits are easy to verify: the peaked string and its corresponding weight serves as a witness for efficient verification. }
    \label{fig:1}
\end{figure}
\section{Introduction}\label{sec: intro}

More than four decades after Feynman’s proposal to simulate quantum physics with quantum computers~\cite{feynman2018simulating}, quantum hardware has progressed rapidly while theory has delivered powerful techniques for quantum error correction and mitigation~\cite{lidar2013quantum,temme2017error,acharya2023suppressing,google2024belowthreshold,bausch2024learning,dasu2025breaking,bluvstein2024logical,zhang2023quantum}. Thanks to these advances, quantum computing today is approaching regimes of classically intractable computation~\cite{harrow2017quantum, preskill2018quantum,kim2023evidence,arute2019quantum,zhong2020quantum,madsen2022quantum,deng2023gbs255,bravyi2018quantum,kretschmer2025demonstrating}.
Among the many tasks that a quantum machine is hoped to outperform its classical counterpart, a prominent near-term target is sampling-based quantum advantage (sometimes called “quantum supremacy”): generating samples from a distribution produced by a quantum circuit for which classical approximation is provably hard. Over the past decade, sampling from random quantum or linear-optical circuits has emerged as a leading candidate, supported by worst-to-average-case complexity evidence for the classical simulability of output probabilities~\cite{aaronson2011computational,aaronson2013bosonsampling,aaronson2016complexity,movassagh2023hardness}. Experimentally, several platforms have reported large-scale demonstrations, including random circuit sampling with superconducting and trapped-ion processors and Gaussian boson sampling with photonic devices~\cite{arute2019quantum,zhong2020quantum,madsen2022quantum,deng2023gbs255,liu2025certified}, while concurrent advances in classical simulation continue to probe the depth and noise regimes where a decisive advantage persists~\cite{pan2022solving,dellios2023simulating,huang2020classical,aharonov2022noisyrcs,zlokapa2023boundaries,morvan2024phase,zhao2024leapfrogging,oh2023spoofing}.

At a high level, an ideal quantum advantage protocol should satisfy all three of the following criteria:
(i) feasibility on near-term noisy hardware,
(ii) average-case classical hardness guarantees, and
(iii) an efficient, scalable verifier.
A persistent obstacle for sampling proposals is (iii) verifiability. Cross-entropy benchmarking for RCS~\cite{boixo2018characterizing,arute2019quantum,aaronson2019classical}, for example, offers a pragmatic test but becomes computationally demanding for $\gtrsim 70$ qubits, and can blur the hardness line when realistic noise and approximate simulators are taken into account~\cite{zlokapa2023boundaries,morvan2024phase}. Recent progress on verification protocols such as Bell sample extraction~\cite{hangleiter2024bell} provide a promising way to certify certain physical properties, yet they do not currently yield rigorous complexity guarantees for the \emph{entire} output distribution.

Outside of sampling, leading advantage candidates typically miss at least one pillar. For instance, the quantum approximate optimization algorithm (QAOA) lacks a general theoretical guarantee of speedup on natural problem families (albeit empirical and heuristic advantages have been observed)~\cite{farhi2014quantum,farhi2016quantum,zhou2020quantum,zhang2021qed,ebadi2022quantum}; instantaneous quantum polynomial (IQP) schemes~\cite{shepherd2009temporally,bremner2011classical,bremner2017achieving,codsi2022classically,maslov2024fast,rajakumar2024polynomial} face ongoing challenges in scalable, robust verification~\cite{kahanamoku2019forging,bremner2023iqp} and from secret-extraction attacks~\cite{gross2023secret}. Shor’s algorithm~\cite{shor1994algorithms}, while offering strong asymptotic guarantees, is not feasible on near-term devices due to the error-correction overhead required for the necessary circuit sizes and depths. Recent demonstrations of beyond-classical behavior in specific quantum simulations are valuable, practical milestones~\cite{daley2022practical, kim2023evidence,king2025beyond,haghshenas2025digital}. However, they differ in emphasis from sampling-based advantage protocols: they typically rely on error-mitigation pipelines or model-specific analyses and may lack average-case hardness guarantees from a computational perspective.
\begin{table}[t]
\centering
\caption{Comparing leading quantum advantage proposals and their known complexity-theoretic properties. Expanded on~\cite{bouland2018quantum}.}
\label{tab:bfnv-proposals}
\begin{tabular}{lccccc}
\toprule
\textbf{Proposal} & \textbf{Worst-case} & \textbf{Average} & \textbf{Efficient}  &\textbf{Near-term} \\
& \textbf{hardness} & \textbf{case hardness} & \textbf{Verification} &  \textbf{Feasibility} \\
\midrule
BosonSampling~\cite{aaronson2013bosonsampling}        & \cmark & \cmark &                  & \cmark         \\
FourierSampling~\cite{fefferman2015power}       & \cmark & \cmark &                 &         \\
IQP~\cite{shepherd2009temporally}                 & \cmark &        &          &         \\
Random Circuit Sampling\cite{aaronson2016complexity,bouland2018quantum}& \cmark & \cmark &          & \cmark  \\
QAOA~\cite{farhi2014quantum}& \cmark & &  \cmark        & \cmark  \\
Peakedness Estimation \cite{aaronson2024verifiable}, our work& \cmark & \cmark & \cmark        & \cmark  \\
\bottomrule
\end{tabular}

\vspace{0.5ex}
\raggedright

\end{table}
 These gaps motivate new sampling-based models whose circuit decomposition remain ``random-looking'' yet contain a simple structure enabling efficient checks~\cite{aaronson2024verifiable}. The proposal of~\cite{aaronson2024verifiable} constructs random peaked circuits (RPCs) that look random globally yet place anomalously large probability on a designated bit string, supplying a simple verification witness through the peaked string without sacrificing the randomness features that underlie hardness arguments. 
\begin{definition}\label{def: peaked}
For a unitary $U$ on $n$ qubits, define the peak weight
\begin{equation}
p_{\max}(U)\;:=\;\max_{x\in\{0,1\}^n}\;|\langle x|U|0^n\rangle|^2.
\end{equation}
A circuit is called $\delta$-peaked if $p_{\max}(U)\geq \delta$. 
A circuit is considered `peaked' for short if $\delta = 1/\poly (n)$.
\end{definition} 
Naturally, calculating the secret string $x$ and its peakedness, $p_x$ can be used as a challenge for verifiable quantum advantage demonstrations. This raises three questions: what are the circuit complexity of these circuits, how hard to classically compute their output distributions are on average, and how to generate such circuits efficiently. We address these questions with analytical and numerical arguments.

To begin with, we give an explicit construction of random peaked circuits closely aligned with the numerical optimization model proposed in~\cite{aaronson2024verifiable}: there, $\tau_C$ layers of random circuits are applied, followed by $\tau_{C'}$ layers of variational circuits, whose parameters are varied with gradient-based search to maximize on some designated output string. For the analytical model considered in this work, we generate random peaked circuits through the following postselection procedure:
\begin{definition}[Random Peaked Circuit (RPC) Construction via Postselection.]\label{def: pcc}
To generate a distribution over “peaked” circuits, consider the following procedure:
\begin{enumerate}
    \item Draw a depth-$\mathrm{poly}(n)$ circuit $C$ from an ensemble forming a unitary $k$-design: e.g. a polynomial-sized random circuit with each gate drawn $\sim \mu$, the Haar random distribution.
    \item Draw a second circuit $C'$ from the same architecture.
    \item Define $P \;=\; C'^{\dagger}C$~\footnote{Throughout the work, we choose to break $P$ into two chunks as a proof strategy, although many of the analytical results will hold when $P$ is picked at random as a whole}. Output $P$ if $P$ is $\delta$-peaked given some desired peakedness $\delta$ and desired peaked string $x_\star$; otherwise, repeat the last procedure.
\end{enumerate}
\end{definition}
Call the distribution of peaked circuits generated this way $\nu_\delta$; or simply $\nu$ whenever there is no ambiguity. Under this construction, we first analytically show that $P$ is peaked on one of its dimensions while other $d-1$ levels remain $k$-design-like, and using a combination of block decomposition, unitary design, we show that the ensemble of $\{P\}$ contains a large amount of distinct circuits, with nearly equal probability mass distribution. Next, from a circuit packing theorem, we show that this random peaked circuit ensemble requires a large number of elementary operations to implement, namely: random peaked circuits are incompressible. We formally cast this result as a lower bound in circuit complexity: Given a fixed gate set, the circuit complexity measures the minimum gates required to prepare a given quantum operation. A direct corollary is that $C'^{\dagger}$ cannot be a trivial or even obfuscated reverse of $C$: this resolves an open problem posed by Aaronson in 2022~\cite{aaronson2022much}. 

\begin{theorem}[Circuit complexity of random peaked circuits. Informal.]\label{thm: incom}
   Suppose we have some architecture with poly-sized random gates that forms a unitary $k$-design. The obtained $C'^\dagger C$ according to \cref{def: pcc} at least requires $\tilde \Omega(nk)$ gates to implement with overwhelmingly high probability. 
\end{theorem}
\begin{figure}
    \centering
    \includegraphics[width=1\linewidth]{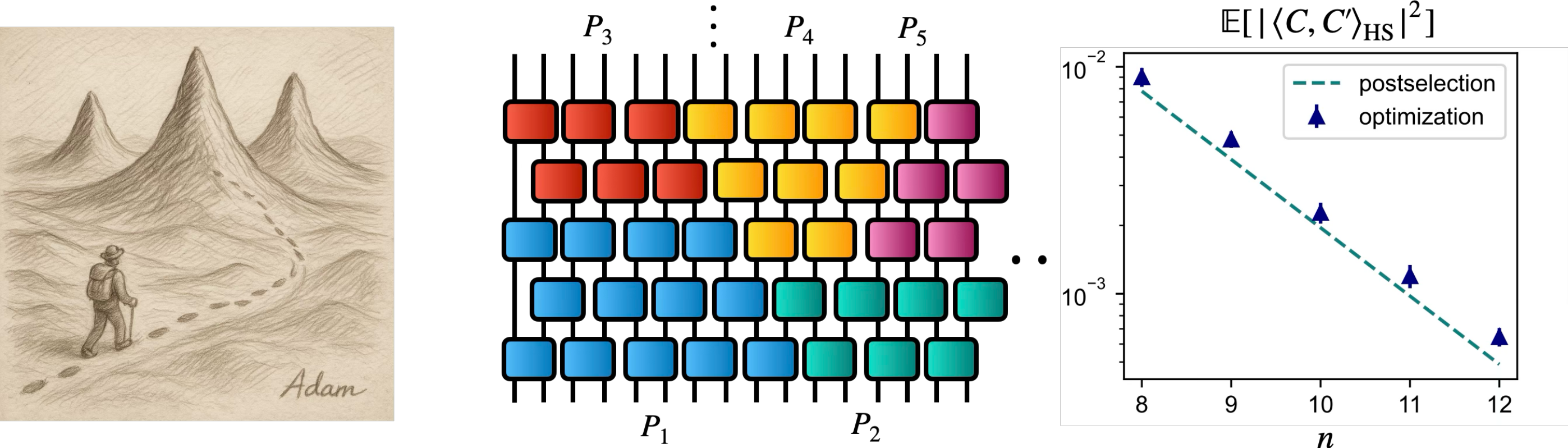}
    \caption{{\bf Left:} In practice, generating RPC from postselection has an extremely small success probability.  Therefore, we randomly generate $C$ and variationally optimize $C'$ with randomized initializations and gradient-based optimizer, Adam. This figure reflects an artist's impression of the process: Adam~\cite{kingma2014adam}, the tireless optimizer start with a random location in the parameter space. Following the local gradient in the landscape, Adam find the closest peaked circuit to the initialization. This search process is a practical rescue for generating random peaked circuits. {\bf Middle:} To construct larger peaked circuits that goes beyond the size of classical simulability, we show that it is possible to combine small peaked circuits while still reliably keep tracking of the peaked string and peakedness. {\bf Right:} Through a comparision in the Hilbert-Schmidt overlap between $C$ and $C'$, the plot shows that the circuit generated with Adam optimization has similar properties to those generated by postselection (as predicted by our theory). The numerical result is averaged over 100 random instances of $C$. This clearly shows that, on average, the random peaked circuits generated numerically are far from identity and very obfuscated.}
    \label{fig:2}
\end{figure}

Secondly, as in the Boson sampling and RQC~\cite{aaronson2011computational,aaronson2013bosonsampling,aaronson2016complexity,bouland2018quantum,movassagh2023hardness} case we are interested in the rigorous complexity result of RPCs. To this end, we first show that \emph{worst case} of nearly exact computation of a single output weight of a peaked circuit (also known as `strong simulation' in literature) is \#P hard. Next, we slightly perturb each gate in the RPC distribution with gates from a hard circuit instance, and we show that the resultant distribution is also peaked. Lastly, we prove that if this slightly perturbed RPC distribution were average case easy, then there exist a low degree polynomial one could construct computationally efficient, whose extraction gives the almost exact estimation of the worst case circuit. Therefore, we reach the second main theorem in this work
\begin{theorem}[Hardness of simulating random peaked circuits. Informal.]\label{thm: hard}
    Estimating $p_{\max}(P)$ for a $P$ sampled with a slightly perturbed version of $\nu$ to within $2^{-\mathrm{poly}(n)}$ additive error is average-case \#P-hard. This result holds even when the peaked string $x_\star$ is known. 
\end{theorem}
Namely, the average case hardness for simulating RPC is as hard as simulating RQC as far as an exponentially small additive error is considered. At inverse–polynomial additive accuracy, the associated worst-case peak-estimation task is PromiseBQP-complete. Beyond worst case, we also rule out a broad classical strategy: we identify a regime that any sequential simulator $\mathcal{S}$ that updates an approximation to quantum state on a gate-by-gate basis cannot approximate the output weight of a $\delta$-peaked circuit $P$ on an average case basis. The intuition is that our construction first routes the computation through a scrambling segment (e.g., $C\ket{0}$ forms a state $k$-design) before $C'$ recovers the peak; maintaining inverse-polynomial error across this stage forces $\mathcal{S}$ to track a Haar-like intermediate state with near-unit fidelity, which is widely believed to be hard. Together with the fact that RPCs are incompressible, this suggests there is no substantially “smarter” classical approach that circumvents this barrier than a sequential gate-by-gate simulation.


 The literal postselection recipe for generating peaked circuits is not scalable in $n$. In a Hilbert space of dimension $d=2^n$, forcing a Haar–random column $\ket{c'}=C'\ket{0^n}$ to be close to a fixed target $\ket{c}=C\ket{0^n}$, namely, to have $|\!\braket{c}{c'}\!|^2 \ge \delta$, occurs with probability exactly $(1-\delta)^{d-1}$, so even modest alignment is exponentially rare in $d$. A similar conclusion holds when $C'$ is drawn from a unitary $k$-design. Consequently, naive postselection has exponentially small acceptance and is impractical at scale.

 We therefore move the difficulty to instance generation and adopt a practical remedy: we variationally synthesize $C'$ so that its first column approximates $\ket{c}$ to fidelity $\delta$. Once such circuits are found, the verifiable quantum-advantage protocol follows directly. Under reasonable assumptions, we argue that this search-based construction preserves the information-theoretic content of the postselected definition while enabling implementation on hardware; moreover, with randomized initialization the optimizer returns a representative among many degenerate realizations, effectively obfuscating the underlying circuit decomposition.
 
 Numerically, averaging over $100$ instances and multiple system sizes with both $C$ and $C'$ taken as brickwall circuits of equal depth $\tau_C=\tau_{C'}=n$, we evaluate the normalized Hilbert–Schmidt overlap $\mathbb{E} [|\langle C,C'\rangle_{\mathrm{HS}}|^2]$ with $d=2^n$ and observe that its mean concentrates near $2/d=2^{\,1-n}$, i.e., it decays exponentially with $n$, just as the predicted property for the RPCs constructed from postselection. 
 
 To scale up, we stitch together small peaked blocks while tracking a designated input–output basis string through each block: if block $i$ is $\delta_i$-peaked on $x_{i-1}\!\to x_i$ (assuming $\delta_i$ is sufficiently close to 1), then the $L$-composed circuit $U$ remains peaked along the tracked path with $|\!\braket{x_L}{U|x_0}\!|^2\approx\prod_i(1-\delta_i)$: choosing $\delta_i = O(1/L)$ makes the peakedness of the composed circuit constant. For a classical challenger without knowledge of the stitching pattern, spoofing remains hard: even in a $1$D layout, $L = O(\log n)$ blocks already induce a superpolynomial number of candidate stitch patterns (e.g., $n^{\Omega(\log n)}$), making brute-force enumeration and testing computationally infeasible. Furthermore, we may apply local circuit rewrite rules to make the stitching pattern more obscure, such that the circuit avoids returning to a concentrated distribution near the boundary between the blocks. We propose using these composed large peaked circuits as a verifiable quantum advantage protocol.

\section{Circuit complexity of random peaked circuits}~\label{sec: complexity}
It is natural to expect peaked circuits to have high circuit complexity: most textbook quantum algorithms that yield sharply peaked output distributions already require polynomially many gates. In this work, however, we move beyond this worst-case intuition and ask an average-case question: what is the typical circuit complexity of a random peaked circuit? We argue that simulating peaked circuits remains computationally hard on average. 
\subsection{Constructing RPC's with postselection}
Throughout the rest of the work, we make use of the following random peaked circuit construction in Def.~\ref{def: pcc}: Pick a random “scrambling” circuit $C$, pick another random circuit $C'$ with the same layout, and set $P=C'^{\dagger}C$. Keep $P$ only if it is $\delta$-peaked on a chosen bitstring $x_\star$; otherwise, resample $C'$ (and/or $C$) and repeat. The key observation is that a randomly drawn peaked circuit, constructed via postselection, has high circuit complexity. 

Without loss of generality and clarity, we set $x_\star=0^n$ (as different computational basis are linked by at most one layer of X gates only) and consider the postselection criteria where the peakedness is exactly 1 throughout this section. Under this postselection, $C'$ and $C$ share the same initial state vector up to a phase, while the action of $C'$ on any other computational basis state remains Haar random (up to normalization). Thus, although $C'$ and $C$ are correlated through their first column, the remaining matrix elements of $C'$ (subject to orthogonality and normalization constraints) are otherwise distributed as in the random ensemble. This construction enables us to study quantum circuits that are highly “peaked” on the $|0^n\rangle$ state, while retaining typical randomness elsewhere, and is useful for analyzing both average- and worst-case properties of such circuits.

One might first worry that the circuits generated this way are merely trivial inverses of each other (i.e., $C' = C$)~\cite{aaronson2024verifiable}. We show here that it is not the case using random circuits' design properties. In fact, we rigorously prove that there is no easy way to `cancel out' the gates in the middle with any circuit rewrites and contractions. In other words:

Here’s the intuition why this is true. If we sample two independent unitaries $C,C'\in\U(2^n)$ and postselect on having the same first column, i.e., $C\ket{0^n}=C'\ket{0^n}=\ket{\psi_0}$, then we’ve fixed a single rank-1 direction $\ket{\psi_0}$ while leaving the action on $\ket{\psi_0}^\perp$ essentially unconstrained. In the Haar case this factorizes \emph{exactly}: there is a change of basis $R$ with $R\ket{0^n}=\ket{\psi_0}$ such that $C$ and $C'$ look like independent blocks on $\U(2^n\!-\!1)$, and the “peaked” unitary $P=C'^\dagger C$ reduces to $\,\mathrm{diag}(1,V)$ with $V\sim$ Haar on $\U(2^n\!-\!1)$, independent of $\ket{\psi_0}$. For unitary $k$-designs the same picture holds up to degree-$k$ moments: conditioning only fixes the first column projectively, and all balanced degree-$\le k$ statistics of $P$ match those of the block model with $V$ drawn from a $(2^n\!-\!1)$-dimensional $k$-design.

This decomposition turns circuit complexity into a problem about the $(d-1)$-block: Haar $V$ has typical $\Theta(4^n)$ two-qubit gate complexity (the $R$ sandwich is just $O(n2^n)$ overhead), and for $k$-design $V$ we get a clean packing to complexity route—there are $\gtrsim (d-1)^k/k!$ well-separated outputs, forcing average (and w.h.p.) lower bounds of $\tilde\Omega(k n)$ gates at constant precision for the peaked ensemble. We provide a formal proof below.

\subsection{Block decomposition under first-column conditioning}

Let $\U(d)$ be the unitary group with Haar probability measure $\mu_d$ and $d=2^n$.
For a unit vector $\ket{\psi_0}\in\C^d$, write its stabilizer
\[
  \mathrm{Stab}(\ket{\psi_0}):=\{W\in \U(d): W\ket{\psi_0}=\ket{\psi_0}\}\cong \U(d-1)\times \U(1).
\]
Fix any $R\in\U(d)$ with $R\ket{0^n}=\ket{\psi_0}$. In the orthogonal decomposition $\C^d=\mathrm{span}\{\ket{\psi_0}\}\oplus \ket{\psi_0}^\perp$ we use block notation
\[
  R^\dagger U R
  =\begin{pmatrix} \alpha & a^\dagger \\[2pt] b & X \end{pmatrix},
  \qquad \alpha\in\C,\ a,b\in\C^{d-1},\ X\in\C^{(d-1)\times(d-1)}.
\]
Throughout, we show that conditioning on equality of first columns can be understood via a regular conditional distribution.

\begin{theorem}[Exact block--Haar decomposition]
\label{thm:haar-block}
Let $C,C'\overset{\mathrm{i.i.d.}}{\sim}\mu$. Condition on $\{C\ket{0^n}=C'\ket{0^n}\}$ and denote $\ket{\psi_0}:=C\ket{0^n}=C'\ket{0^n}$.
Then there exist independent random variables
\[
  \phi,\phi'\sim \text{Unif}[0,2\pi),\qquad X,Y\overset{\mathrm{i.i.d.}}{\sim}\mu_{d-1}
\]
such that, for any (measurable) choice of $R\in\U(d)$ with $R\ket{0^n}=\ket{\psi_0}$,
\[
  C \ \overset{d}{=}\ R \begin{pmatrix} e^{i\phi} & 0 \\[2pt] 0 & X \end{pmatrix},
  \qquad
  C' \ \overset{d}{=}\ R \begin{pmatrix} e^{i\phi'} & 0 \\[2pt] 0 & Y \end{pmatrix}.
\]
Consequently,
\[
  P:=C'^{\dagger}C \ \overset{d}{=}\ \,\diag\!\bigl(1,\,V\bigr)\,,
  \qquad V:=Y^\dagger X \sim \mu_{d-1},
\]
and $V$ is independent of $\ket{\psi_0}$ (and hence of $R$ aside from the conjugation).
\end{theorem}

\begin{proof}
Let $\pi:\U(d)\mapsto \mathbb{S}^{2d-1}$ be $\pi(U)=U\ket{0^n}$. The pushforward $\pi_\#\mu$ is uniform on the sphere (left invariance). By disintegration there is a regular conditional law $\mu(\cdot\,|\,\ket{\psi_0})$ on each fiber $\pi^{-1}(\ket{\psi_0})$.
Right-multiplication by $\mathrm{Stab}(\ket{\psi_0})$ preserves the fiber, and $\mu(\cdot\,|\,\ket{\psi_0})$ is right-invariant under $\mathrm{Stab}(\ket{\psi_0})$.
Transport by $R^\dagger$ identifies the fiber with $\mathrm{Stab}(\ket{0^n})=\{\diag(e^{i\phi},Z):\phi\in\R,\,Z\in\U(d-1)\}$ and sends $\mu(\cdot\,|\,\ket{\psi_0})$ to product Haar on $\U(1)\times \U(d-1)$. The claims follow; $Y^\dagger X\sim\mu_{d-1}$ by left/right invariance and is independent of $\ket{\psi_0}$ since $X,Y$ live on $\ket{\psi_0}^\perp$.
\end{proof}

\begin{corollary}[Peaked ensemble under Haar sampling]
\label{cor:peaked-haar}
If $C,C'\overset{\mathrm{i.i.d.}}{\sim}\mu$ and we condition on $C\ket{0^n}=C'\ket{0^n}$, then
$P=C'^{\dagger}C \overset{d}{=} \,\diag(1,V)\,$ where $V\sim\mu_{d-1}$.
\end{corollary}

\medskip

We now state the moment-equivalence for unitary designs. A \emph{unitary $k$-design} on $\U(d)$ is a distribution such that the expectation of every balanced polynomial in the entries of $U$ and $\overline{U}$ of total degree $\le k$ matches Haar (for $\eta$-approximate designs, moments match up to $O(\eta)$).

\begin{proposition}[Block--design decomposition up to degree-$k$ moments]
\label{prop:design-block}
Let $C,C'$ be i.i.d.\ from a (possibly $\eta$-approximate) unitary $k$-design on $\U(d)$. Condition on $C\ket{0^n}=C'\ket{0^n}=\ket{\psi_0}$, and fix $R$ with $R\ket{0^n}=\ket{\psi_0}$.
Then for any balanced polynomial $p$ of total degree $\le k$ in the entries of $P=C'^{\dagger}C$ and $\overline{P}$,
\[
  \mathbb{E}\!\big[\,p(P)\,\big]
  \ =\ \mathbb{E}\Big[\,p\big(\,\diag(1,Y^\dagger X)\,\big)\,\Big]\,
\]
where $(\phi,X)$ and $(\phi',Y)$ have the same degree-$\le k$ moments as independent draws from $\U(1)\times\U(d-1)$ (i.e., $X,Y$ are unitary $k$-designs on $\U(d-1)$, moment-independent up to degree $k$).
\end{proposition}

\begin{proof}[Proof sketch]
Pushforward by $U\mapsto U\ket{0^n}$ gives a projective $k$-design on $\CP^{d-1}$. For the fiber $\mathrm{Stab}(\ket{0^n})\cong \U(1)\times\U(d-1)$, the conditional degree-$\le k$ moments are the Haar values. We are also guaranteed independence between the first and rest rows at the level of moments. For $\eta$-approximate designs, this expectation incur an $O(\eta)$ error.
\end{proof}

\begin{corollary}[Peaked ensemble under $k$-design sampling]
\label{cor:peaked-design}
Let $C,C'$ be i.i.d.\ from a unitary $k$-design on $\U(d)$ and condition on $C\ket{0^n}=C'\ket{0^n}$. Then for any degree-$\le k$ balanced polynomial $p$,
\[
  \mathbb{E}\,p\!\big(C'^{\dagger}C\big)
  \ =\ \mathbb{E}\,p\!\big(\,\diag(1,V)\,\big)\ ,
  \qquad V\ \text{a unitary $k$-design on }\U(d-1).
\]
\end{corollary}

\begin{observation}[Normalized Hilbert--Schmidt inner product]
\label{obs:hs}
For $P=\diag(1,V)$ with $V\in\U(d-1)$, the normalized Hilbert--Schmidt inner product is
\[
  \langle C,C'\rangle_{\mathrm{HS}}:=\tfrac1d\,\Tr(C^\dagger C')=\tfrac1d\,\Tr(P)=\tfrac{1+\Tr(V)}{d}.
\]
Thus $\E[\langle C,C'\rangle_{\mathrm{HS}}]=\tfrac1d$ and Hilbert--Schmidt overlap being $\E\!\big[|\langle C,C'\rangle_{\mathrm{HS}}|^2\big]=\tfrac{2}{d^2}$, since $\E[\Tr(V)]=0$ and $\E[|\Tr(V)|^2]=1$ for Haar $V\in\U(d-1)$, and $k\geq2$ designs. 
\end{observation}
One could further use this design property to bound the pair-wise correlation between gates in $C$ and $C'$. We discuss this in the Appx.~\ref{appx: gate}.
\subsection{Packing-based circuit lower bounds for the peaked ensemble}
Since peaked circuits form designs on the non-peaked dimensions, the ensemble naturally contains a large number of distinct circuits that lead to a circuit lower bound to implement via a counting argument. We analytically prove this with circuit packing here.
As previously, we fix $d=2^n$. Let $\|\cdot\|_F$ denote the Frobenius norm, and let $\mathcal{C}_{n,\le s}$ be the family of $n$-qubit circuits with at most $s$ two-qubit gates (allowing arbitrary single-qubit gates for free).

\begin{lemma}[State multiplicity and packing from a $k$-design]\label{lem:state-packing}
Let $V$ be drawn from a unitary $k$-design on $\U(d-1)$ (either exact or $\eta$-approximate). Fix any unit vector $\ket{\psi_0}\in\C^{d-1}$ and define $\ket{\varphi_V}:=V\ket{\psi_0}$.
Then, for any fixed constant $\delta\in(0,1)$, there is a $\delta$-separated subset $\{\ket{\varphi_j}\}_{j=1}^M\subset\CP^{d-2}$ with
\[
  M \ \ge\ c_\delta \binom{d+k-2}{k} \ \gtrsim\ c_\delta \frac{(d-1)^k}{k!}\qquad(d-1\ge k),
\]
where $c_\delta>0$ depends only on $\delta$. For $\eta$-approximate designs the same bound holds up to a multiplicative factor $1-O(\eta)$.
\end{lemma}

\begin{lemma}[Lifting separated states to separated peaked unitaries]\label{lem:lifting}
Fix $\ket{\psi_0}\in\C^d$ with $\|\ket{\psi_0}\|=1$ and $R\in\U(d)$ with $R\ket{0^n}=\ket{\psi_0}$. For each $V\in\U(d-1)$ set $P_V:=\,\diag(1,V)\,$. 
If $\{\ket{\varphi_j}\}_{j=1}^M\subset \ket{\psi_0}^\perp$ is $\delta$-separated, choose $V_j$ with $V_j\ket{\psi_0}=\ket{\varphi_j}$, and define $P_j:=P_{V_j}$. Then
\[
  \|P_i-P_j\|_F \ \ge\ \|\ket{\varphi_i}-\ket{\varphi_j}\|_2 \ \ge\ \delta
  \qquad(i\ne j).
\]
\end{lemma}

\begin{proposition}[Covering number of $s$-gate circuits (layout-agnostic)]\label{prop:covering}
There exist universal constants $C,\kappa>0$ such that for every $\varepsilon\in(0,1/2]$,
\[
  \mathcal{N}\!\bigl(\varepsilon;\,\mathcal{C}_{n,\le s},\,\|\cdot\|_F\bigr)
  \ \le\ \Bigl(\tfrac{C\,n^2\,s}{\varepsilon}\Bigr)^{\kappa s}.
\]
Consequently, any $\varepsilon$-separated subset of $\mathcal{C}_{n,\le s}$ has size at most $(C n^2 s/\varepsilon)^{\kappa s}$.
\end{proposition}

As a remark: If the two-qubit interaction pattern is fixed (e.g., brickwork on a line), the $n^2$ factor can be absorbed into constants, yielding $\mathcal{N}(\varepsilon)\le (C s/\varepsilon)^{\kappa s}$. Next, we turn this counting into a circuit lower bound by comparing the number of `unique' circuits in the peaked ensemble and the number of circuits reachable using $s$ gates

\begin{theorem}[Packing into circuit lower bound]\label{thm:packing-lb}
Let $P$ be drawn from the peaked ensemble $P=\,\diag(1,V)\,$ where $R\ket{0^n}$ is (projectively) Haar and $V$ is drawn from a unitary $k$-design on $\U(d-1)$, independently of $\ket{\psi_0}$. 
Fix constants $\delta\in(0,1)$ and $\varepsilon\in(0,\delta/3]$. Then
\[
  \Pr\!\Big[\ \dist\big(P,\mathcal{C}_{n,\le s}\big)\le\varepsilon\ \Big]
  \ \le\
  \frac{(C n^2 s/\varepsilon)^{\kappa s}}{\,c_\delta\,\binom{d+k-2}{k}}\;+\;O(\varepsilon).
\]
In particular, if
\[
  (C n^2 s/\varepsilon)^{\kappa s}\ \le\ \tfrac12\,c_\delta\,\binom{d+k-2}{k},
\]
then at least a $1/2$-fraction of the ensemble requires more than $s$ two-qubit gates to achieve $\varepsilon$ accuracy.
\end{theorem}
From this counting argument, two corollaries follow: the first is an an average case circuit complexity lower bound, and the second bounds the `typical case' circuit complexity of a random peaked circuit.
\begin{corollary}[Average-case lower bound]\label{cor:avg-lb}
For $d=2^n$ and any fixed $\varepsilon\in(0,1/10]$,
\[
  \mathbb{E}\big[\,\mathrm{GateCount}_\varepsilon(P)\,\big]\ \ge\ \Omega\!\Big(\frac{k\,n}{\log(k n)}\Big),
\]
for the peaked ensemble built from a unitary $k$-design on $\U(d-1)$.
\end{corollary}
\begin{proof}
Fix $\delta=1/3$ and any constant $\varepsilon\in(0,\delta/3]$. Let
\[
s_\star := \big\lfloor \alpha\,\tfrac{k\,n}{\log(k n)}\big\rfloor
\]
for a sufficiently small universal constant $\alpha>0$. By Theorem~\ref{thm:packing-lb},
\[
(C n^2 s_\star/\varepsilon)^{\kappa s_\star}\ \le\ \tfrac12\,c_\delta\,\binom{d+k-2}{k}
\]
for all large $n,k$, hence at least a $1/2$-fraction of the ensemble requires more than $s_\star$ two-qubit gates to achieve $\varepsilon$ accuracy. Therefore
\[
\mathbb{E}\big[\mathrm{GateCount}_\varepsilon(P)\big]
\ \ge\ s_\star\cdot \Pr\!\big[\mathrm{GateCount}_\varepsilon(P)\ge s_\star\big]
\ \ge\ \tfrac12\,s_\star
\ =\ \Omega\!\Big(\tfrac{k\,n}{\log(k n)}\Big).
\]
Alternatively, applying Lemma~\ref{lem:whp} with $s=c_0\,\tfrac{k n}{\log(k n)}$ gives
$\Pr[\mathrm{GateCount}_\varepsilon(P)\le s]\le e^{-c_1 k n}$ and hence
$\mathbb{E}[\mathrm{GateCount}_\varepsilon(P)]\ge s(1-e^{-c_1 k n})
=\Omega\!\big(\tfrac{k n}{\log(k n)}\big)$.
\end{proof}
\begin{corollary}[High-probability lower bound]\label{lem:whp}
There exist universal constants $c_0,c_1,c_2,\varepsilon_0>0$ such that the following holds.
Let $\varepsilon\in(0,\varepsilon_0]$ and $s\le c_0\,\dfrac{k\,n}{\log(k n)}$. For the peaked ensemble from an (exact or sufficiently accurate) unitary $k$-design on $\U(d-1)$,
\[
  \Pr\!\Big[\ \mathrm{GateCount}_\varepsilon(P)\ \le\ s\ \Big]
  \ \le\ \exp\!\big(-c_1\,k\,n\big),
\]
hence
\[
  \Pr\!\Big[\ \mathrm{GateCount}_\varepsilon(P)\ \ge\ c_2\,\tfrac{k\,n}{\log(k n)}\ \Big]\ \ge\ 1-\exp\!\big(-c_1\,k\,n\big).
\]
\end{corollary}
\begin{proof}
Fix $\delta=1/3$ and a constant $\varepsilon\in(0,\delta/3]$. By Theorem~\ref{thm:packing-lb},
\[
\Pr\!\big[\mathrm{dist}(P,\mathcal C_{n,\le s})\le\varepsilon\big]
\;\le\;
\frac{(C n^2 s/\varepsilon)^{\kappa s}}{c_\delta\,\binom{d+k-2}{k}}+O(\varepsilon).
\]
Using $\binom{d+k-2}{k}\ge (d-1)^k/k!$ and $d=2^n$, the exponent of the RHS is
\[
\kappa s\log(C n^2 s/\varepsilon)\;-\;k\log(d-1)\;+\;O(k\log k)
= \kappa s\log(C n^2 s/\varepsilon)\;-\;\Theta(kn)\;+\;O(k\log k).
\]
If $s\le c_0\,\frac{kn}{\log(kn)}$ with $c_0$ small enough, the exponent is $\le -c_1 kn$ for universal $c_1>0$, so the RHS is $\le e^{-c_1 kn}$ (absorbing the $O(\varepsilon)$ term into constants by taking $\varepsilon$ fixed). This gives
\[
\Pr\!\big[\mathrm{GateCount}_\varepsilon(P)\le s\big]\;\le\;e^{-c_1 kn},
\]
which is the claim. \qedhere
\end{proof}
Meanwhile, when the unitary circuits are drawn from truly Haar-random distribution instead of designs, the following theorem can be proved similarly:
\begin{theorem}[Haar-scale lower bound for the peaked Haar ensemble]\label{cor:haar-4n}
Let $P=\,\diag(1,V)\,$ with $V\sim\mu_{d-1}$ (Haar). For any fixed $\varepsilon\in(0,1/10]$ there exist universal constants $c_3,c_4>0$ such that, with probability at least $1-\exp\!\big(-c_3 d^2\big)$,
\[
  \mathrm{GateCount}_\varepsilon(P)\ \ge\ c_4\,d^2 \ =\ c_4\,4^n.
\]
\end{theorem}
\begin{proof}
Fix $\varepsilon\in(0,1/10]$ and write $N:=d-1$. By a standard volumetric
packing bound on compact Lie groups (bi-invariant metric from $\|\cdot\|_F$),
there exists a $\tfrac{\varepsilon}{2}$-separated set
$\{V_j\}_{j=1}^M\subset \U(N)$ with
\[
  M \ \ge\ \Big(\tfrac{c}{\varepsilon}\Big)^{N^2}
\]
for a universal constant $c>0$. Define $P_j:=R\,\diag(1,V_j)\,R^\dagger$.
Conjugation preserves $\|\cdot\|_F$, and
$\|\diag(1,V_i)-\diag(1,V_j)\|_F=\|V_i-V_j\|_F$, so
$\{P_j\}$ is $\tfrac{\varepsilon}{2}$-separated as well.

Let $\mathcal{C}_{n,\le s}$ be $n$-qubit circuits with $\le s$ two-qubit gates.
By Proposition~\ref{prop:covering} (layout-agnostic form),
there exists an $\varepsilon/2$-net of $\mathcal{C}_{n,\le s}$ of size at most
$(C n^2 s/\varepsilon)^{\kappa s}$. For Haar $V$ (hence Haar $P$ on the peaked
manifold), bi-invariance implies that the measure of any $\varepsilon/2$-ball is
at most $1/M$. Therefore
\[
  \Pr\!\big[\dist(P,\mathcal{C}_{n,\le s})\le \varepsilon\big]
  \ \le\ \frac{(C n^2 s/\varepsilon)^{\kappa s}}{M}
  \ \le\ \exp\!\Big(\kappa s\log(C n^2 s/\varepsilon)\;-\;N^2\log(c/\varepsilon)\Big).
\]
Choose $s=\alpha d^2$ with $\alpha>0$ small enough (depending only on
$C,\kappa,c,\varepsilon$). Since $N^2=(d-1)^2=d^2+O(d)$,
the exponent is $\le -c_3 d^2$ for some universal $c_3>0$, hence
\[
  \Pr\!\big[\mathrm{GateCount}_\varepsilon(P)\le \alpha d^2\big]\ \le\ e^{-c_3 d^2}.
\]
Renaming $c_4:=\alpha$ yields the claim:
with probability at least $1-e^{-c_3 d^2}$,
$\mathrm{GateCount}_\varepsilon(P)\ge c_4 d^2=c_4 4^n$.
\end{proof}

All results above remain valid for $\eta$-approximate unitary $k$-designs on $\U(d-1)$, up to replacing $c_\delta$ by $c_\delta(1-O(\eta))$ in Lemma~\ref{lem:state-packing} and absorbing an additive $O(\eta)$ in the probability bounds of Theorem~\ref{thm:packing-lb}. In Appx.~\ref{sec: alter} we provide an alternative proof for the same incompressibility theorem.

\section{Hardness of simulating peaked circuits}~\label{sec: hardness}
In the previous section we proved that \emph{random peaked circuits} (RPCs) obtained by postselecting from a $k$-design requires at least circuit complexity \emph{linear} in $k$ to implement. Since local random circuits generate approximate unitary $k$-designs in shallow depth $O(\log n \cdot k)$~\cite{schuster2025random}, this gives a tight bound in $k$ dependence for RPCs and establishes the near-term feasibility of preparing them on current devices.

At the same time, the linear dependence in $k$ scaling is a necessary (though not obviously sufficient) condition for classical simulation to be hard in the regime of interest: picking $k = \log(n)$ already make naive simulation algorithms hard.

To leverage RPCs for \emph{verifiable quantum advantage}, we therefore make the target task explicit and focus on \emph{strong simulation}: given a circuit $P$ and bit string $x\in\{0,1\}^n$, estimate the single-output weight
\[
p_x(P)\;:=\;|\!\braket{x}{P\,|0^n}\!|^2
\]
to prescribed precision (e.g., exponentially small additive error). To this end, we show that simulating RPCs is, both in worst case and average case, as hard as simulating RQCs.
\subsection{`Almost exact' simulation is \#P Hard}
\subsubsection{Worst case hardness via a peaked embedding}
First, how hard is it to simulate peaked circuits in the worst case? 
We first observe that computing a single output probabilities exactly is \#P-complete for generic random quantum circuits, and this hardness carries over to peaked circuits, as any circuit can be converted into a peaked circuit with adding merely one ancilla:
Let $C$ be any $n$-qubit circuit drawn from a polynomial-time samplable distribution over universal gate sets, and define the peaked embedding on one ancilla:
\[
P\;:=\; \big(\ket{0}\!\bra{0}\otimes I \;+\; \ket{1}\!\bra{1}\otimes C\big)\,(H\otimes I),
\qquad\text{acting on }\ket{0}\ket{0^n}.
\]
Measuring in the computational basis gives
\[
\Pr[0,0^n]=\tfrac{1}{2},\qquad 
\Pr[1,x]=\tfrac{1}{2}\,p_x(C),\quad
\Pr[0,x\neq 0^n]=0,
\]
so conditioned on the ancilla outcome $1$ we recover $p_x(C)$ exactly.

\begin{lemma}[Strong-simulation hardness transfers]
\label{prop:peak-strong}
There exist a polynomial-sized peaked circuit $P$ and target string $x$ such that exactly evaluating $p_x(P)$ is \#P-hard
\end{lemma}

\begin{proof}[Proof sketch]
For any $x\in\{0,1\}^n$,
\[
p_x(C)\;=\;2\,p_{1,x}(P).
\]
Hence, exactly computing $p_{1,x}(P)$ is polynomial-time equivalent to exactly computing $p_x(C)$. In particular, if computing $p_x(C)$ is \#P-hard (worst case or average case over the ensemble of $C$), then computing $p_{1,x}(P)$ is also \#P-hard. The controlled implementation increases size/depth by only a constant factor.

The displayed distribution is immediate from linearity. The identity $p_x(C)=2\,p_{1,x}(P)$ yields an equivalence for exact probability computation. Known results give \#P-hardness of exactly computing $p_x(C)$ for worst-case circuits. The controlled-$C$ can be realized by adding a control to each gate of $C$, incurring a constant-factor overhead.
\end{proof}
\subsubsection{Worst-to-average reduction with a polynomial method}
Furthermore, for random-circuit ensembles there is a worst-to-average-case reduction establishing that computing typical output probabilities is also \#P-hard. This average case hardness gives a strong analytical guarantee for the quantum advantage in RQC sampling. Does this hardness result still hold for RPCs? We give an affirmative answer here.

We prove a worst-to-average reduction with a polynomial method. Remarkably, even when the peaked string is known, estimating its peakedness is still hard on average. W.l.o.g, we set $x_\star = 0^n$. First always assume we are working with an architecture $\mathcal{A}$ with $\text{poly} (n)$ gates such that there exist an instance peaked circuit $P^* =  G_m^* \cdots G_1^*$ whose output amplitude on $x_\star$ is \#P-hard to compute. Next, we embed a tiny bit of knowledge from $P^*$ into the random peaked distribution, $\nu_\delta$, defining a distribution of so-called $\theta$-perturbed peaked circuits, which satisfies $P(0) = P$ (the random peaked circuit) and $P(2\pi) = P^*$ (the hard circuit).

\begin{definition}[$\theta$-perturbed peaked circuits]~\label{def: theta}
Let $P^* = G_m^* \cdots G_1^*$ be a quantum circuit whose output amplitude is \#P-hard to compute (the worst-case instance), and $P = G_m \cdots G_1$ a “peaked” circuit sampled from a distribution $\nu$ defined by \cref{def: pcc}.

Define, for each gate $j$, a Hermitian generator
\[
H_j := \frac{1}{2\pi i} \log(G_j^* G_j^{-1}),
\]
and for $\theta \in [0, 2\pi]$, set
\[
G_j(\theta) := G_j \cdot e^{-i\theta H_j}.
\]

Then, the interpolated circuit is
\[
P(\theta) := G_m(\theta) \cdots G_1(\theta),
\]

The $\theta$-perturbed peaked circuit distribution is defined by sampling $P$ as above, and then forming $P(\theta)$ as above for $\theta \in [0, \Tilde{\theta}]$.
\end{definition}

The polynomial method exploits the fact that, when interpolating between a random “peaked” quantum circuit and a worst-case hard instance, the output amplitude as a function of the interpolation parameter $\theta$ is a low-degree polynomial in $\theta$. If a classical algorithm could efficiently compute output probabilities for most random (peaked) circuits, then, by evaluating the polynomial at enough points and interpolating, it could also efficiently compute the output for the worst-case hard circuit. This establishes a worst-to-average-case reduction: efficient average-case simulation would imply efficient worst-case simulation, which is widely believed to be impossible. To use the polynomial method, we need to cast the exponential from each ``pull back'' gate as a polynomial. Following Bouland et al.~\cite{bouland2018quantum}, we write down the Taylor series expansion of each gate and then truncate:
\begin{definition}[$(\theta, K)$-truncated perturbed random peaked circuit]~\label{def: theta_k}
Using the same definition from \cref{def: theta}, the $(\theta, K)$-truncated perturbed random peaked circuit is then
\[
P^{(K)}(\theta) := G_m^{(K)}(\theta) \cdots G_1^{(K)}(\theta).
\]
where \[
G_j^{(K)}(\theta) := G_j \cdot \sum_i^K ({-i\theta H_j})^i.
\]
By construction, $P^{(K)}(0) = P$ (the original peaked circuit), and as $K \to \infty$, $P^{(K)}(\theta) \to P(\theta)$, the analytic $\theta$-perturbed interpolation between hard and peaked circuits.

\end{definition} 
It would be important to check if the $(\theta, K)$-truncated perturbed RPCs are still peaked:
\begin{fact}[$\theta$-perturbed random peaked circuits are peaked.]
Consider a quantum circuit $C$ composed of $m$ gates acting on $n$ qubits, and define a perturbed circuit $C(\theta)$ in which each gate $G_j$ is replaced by $G_j(\theta) = e^{-i\theta H_j} G_j$, with $H_j$ a Hermitian operator satisfying $\|H_j\| \leq 1$. Let $U = C$ and $U(\theta) = C(\theta)$ denote the unitaries implemented by the original and perturbed circuits, respectively.
\end{fact}
\begin{proof}
It follows from standard operator norm inequalities that
\begin{equation}
\| P - P(\theta) \|_{\mathrm{op}} \leq m |\theta| + O(m\theta^2)
\end{equation}
for sufficiently small $\theta$. Here, the bound accumulates linearly in the number of perturbed gates $m$.

Given any input state $|\psi\rangle$, let $p(x) = |\langle x| U |\psi\rangle|^2$ and $q(x) = |\langle x| P(\theta) |\psi\rangle|^2$ denote the output probability distributions over measurement outcomes $x$ in the computational basis. By standard linear algebra arguments, the total variation distance between $p$ and $q$ is bounded as
\begin{equation}
\| p - q \|_1 \leq 2 \| P - P(\theta) \|_{\mathrm{op}} \leq 2 m |\theta| + O(m\theta^2) \, .
\end{equation}

This shows that for our construction, the output distribution of the $\theta$-perturbed circuit remains close to that of the original circuit for sufficiently small $\theta$. In particular, the distance vanishes linearly with $\theta$ and the number of perturbed gates $m$. It suffices to set $\theta\ll \delta/m$ to retain peakedness in the distribution.
\end{proof}

\begin{fact}[$(\theta, K)$-truncated perturbed RPCs are peaked]
Choosing $1/\tilde\theta \gg m/\delta$ and $K = \text{poly}(n)$. The circuits defined in \cref{def: theta_k} are at least $1/\text{poly}$-peaked.
\end{fact}
This can be proved using a Feynman path integral method and is presented in Bouland et al. paper~\cite{bouland2018quantum}. Compared to the un-truncated circuit, the $(\theta, K)$-truncation gives error at most $\frac{2^{O(mn)}}{(K!)^m}$ on any single output weight. Choosing $k=1/\poly(n)$ suffices to make the difference exponentially small.

\begin{theorem}
    ``Almost exactly simulating'' (that is, computing $p_0(P)$ with error $2^{-\text{poly}(n)}$) for $P$ sampled from $(\theta, K)$-truncated perturbed RPCs with probability $>8/9$ is \#P hard.
\end{theorem}
\begin{proof}
    First observe that there are two randomness in generating $P$ and randomness in picking $\theta$. Assuming there exist some machine $\mathcal{O}$ that compute a random $P^{(K)}(\theta)$ w.p. $2/3$. Then from a counting argument, for at least $2/3$ choice of $P$, $O$ correctly computes the amplitude for $P^{(K)}(\theta)$ w.p. $\geq 2/3$.
    Now let's fix $P$ and fix $k$ distinct evaluation points
   $$
     \theta_1,\dots,\theta_k \;\in\; \Bigl[0,\tfrac1{\mathrm{poly}(n)}\Bigr).
   $$

   Define an oracle $O'$ which, on input $\ell\in\{1,\dots,k\}$, queries $O(\theta_\ell)$. Now $O'$ performs a reconstruction of the polynomial via Berlekamp–Welch. Using the $k$ pairs
   $\bigl\{(\theta_\ell,O(\theta_\ell))\bigr\}_{\ell=1}^k$, $O'$ invokes the Berlekamp–Welch algorithm to recover the unique degree-$d$ polynomial
   $\tilde q$ (with $d=2mK$) that agrees with at least $\tfrac{k+d}2$ of these points. It then outputs $\tilde q(1)$.

\medskip

Take $k = 100\,mK$, by a Markov‐inequality argument, with high probability at least $\tfrac{k+d}2$ of the samples $\{\theta_\ell\}$ land in the “good” regime, so Berlekamp–Welch can be used to recover the true polynomial $q$. 
By assumption, each data point might have some small additive noise. $\tilde q(1)$. From a standard extrapolation amplifier analysis, we know that 
    $$|\tilde q(1)- q(1)| \approx \frac{\varepsilon}{\sqrt{k}} \times \left( \frac{|2 - \tilde\theta|}{\tilde\theta} \right)^{mK}$$
Since $\varepsilon = 2^{-\poly(n)}$, this error will also be $2^{-\poly(n)}$ small. 

 Hence $\tilde q(1) \approx q(1) = p_0\bigl(C'(1)\bigr)$.
 Since at least $2/3$ of the choices of $P$ are “good,” repeating this whole process $O(1)$ times and taking a majority vote still yields $p_0(P^{(K)}(1))$ within a $2^{-\poly(n)}$ additive error.

\end{proof}

Our proof shows that simulating random peaked circuits on average is ``as hard as'' simulating random quantum circuits. The robustness here is not very optimal, as in peaked circuit sampling, one would like the hardness to be $1/\poly(n)$ additive error.
The robustness here can be improved slightly with a noise-robust version~\cite{bouland2022noise} paper, assuming the classical challenger is given access to an NP oracle.
Then, the additive error may be improved to $2^{O(m)}$ using this method, but further improvement seems unlikely. Ultimately, as the range of data is small, this cannot be done because the extrapolation requires some extreme precision. Therefore a completely different proof strategy is desired for proving average hardness with $1/\poly(n)$ additive error. Nevertheless, we show evidence that the currently widely used simulation algorithms should fail.
\subsection{Additive inverse-poly error is PromiseBQP-complete}
\label{subsec:additive-bqp}

\begin{proposition}[Transmission-gap decision]
We fist observe that even relaxing the condition to $1/\poly(n)$ additive error, simulating peaked circuits is still at least BQP-hard, as essentially, every `useful' problem in BQP is a $1/\poly(n)$ peaked circuit.
\label{prop:transmission-gap-add}
Let $U$ be a $\poly(n)$-size quantum circuit on $n$ qubits and define
$p := \big|\!\braket{0^n}{U|0^n}\!\big|^2$. The promise problem of deciding whether
$p \ge 2/3$ or $p \le 1/3$ is \emph{PromiseBQP}-complete (e.g., via mappings from massive $\phi^4$ theory scattering/vacuum probabilities~\cite{jordan2018bqp}).
\end{proposition}

\begin{lemma}[Peaked-weight estimation under inverse-polynomial additive error]
\label{lem:peaked-weight-add}
Let $\epsilon(n)\le 1/\poly(n)$. For any $\poly(n)$-size peaked circuit $U_{\rm peak}$ with designated outcome $y_\star$ and
probability $p_{\max}:=\Pr[y_\star]$, the following estimation problem is \emph{PromiseBQP}-complete: output $\tilde p$ such that
$|\tilde p -p_{\max}| \le \epsilon(n)$.
\end{lemma}
\begin{proof}
\emph{Hardness.} Reduce from Prop.~\ref{prop:transmission-gap-add}. Given $U$ with $p=\big|\!\braket{0^n}{U|0^n}\!\big|^2$, define
\[
U_{\rm peak}:=\bigl(\ketbra{0}{0}\!\otimes I + \ketbra{1}{1}\!\otimes U\bigr)\,(H\otimes I),\quad\text{on }\ket{0}\ket{0^n}.
\]
Measuring in the computational basis yields $p_{\max}:=\Pr[(1,0^n)]=\tfrac12\,p$. Thus $p\in\{\le \tfrac13,\ \ge \tfrac23\}$ iff
$p_{\max}\in\{\le \tfrac16,\ \ge \tfrac13\}$, a constant gap $\Delta=\tfrac{1}{6}$. Any $\epsilon(n)\le 1/\poly(n)\ll \Delta$
lets us decide which side of the gap we are in via a fixed threshold (e.g., $1/4$), proving PromiseBQP-hardness.

\emph{Membership.} A quantum computer can estimate $p_{\max}$ to additive error $\epsilon(n)$ in $\poly(n)$ time by either
(i) direct sampling with $O(1/\epsilon(n)^2)$ shots and Chernoff bounds, or
(ii) amplitude estimation achieving $O(1/\epsilon(n))$ query complexity.
Hence the problem lies in PromiseBQP.

\end{proof}
We now show that, on average, there exists threshold $\delta_\star$ such that for random peaked circuits with $\delta \ge \delta_\star$, random peaked circuits cannot be efficiently simulated by the class of classical \emph{sequential simulators}.

\begin{definition}[Sequential simulator]\label{def:simulator}
A sequential simulator
is any classical algorithm that processes a quantum circuit
$U = G_m \cdots G_1$ gate-by-gate (or layer-by-layer) while maintaining,
after gate $i$, a classical description of an approximate state
$\ket{\psi^{\,i}_{\mathrm{approx}}}\simeq\ket{\psi^{\,i}_{\mathrm{ideal}}}:=G_i\cdots G_1\ket{\psi_{\mathrm{in}}}$.
We assume \emph{monotone fidelity}: for every $i<m$,
\begin{equation}\label{eq:monotone-fid}
\bigl|\braket{\psi^{\,i+1}_{\mathrm{approx}}}{\psi^{\,i+1}_{\mathrm{ideal}}}\bigr|^2
\;\le\;
\bigl|\braket{\psi^{\,i}_{\mathrm{approx}}}{\psi^{\,i}_{\mathrm{ideal}}}\bigr|^2.
\end{equation}
\end{definition}

The monotonic‑fidelity assumption above characterizes a “direct” simulator that updates the state gate‑by‑gate and performs only local truncations. Ordinary floating‑point Schrödinger simulators and the standard TEBD/MPS tensor‑network implementations automatically satisfy this property. For example: 
\begin{itemize}
    \item stabilizer‑rank truncation (discarding small‑weight non‑Clifford components after each gate) and
    \item projected entangled‑pair (PEPS) time evolution with fixed bond dimension, where every truncation step can only decrease the fidelity with the ideal quantum state.
\end{itemize}

The intuition is that the peaked circuit from our construction would first go through an `anti-concentration' phase due to the application of $C\ket{0^n}$, making it extremely hard to keep track of. Specifically, there exist a multiplicative fidelity threshold beyond which the simulator could not reach, due to the following two assumptions from random circuits:

\begin{assumption}[Anti-concentration of RQC]\label{fact:ac}
There exist constants $\alpha,\beta>0$ such that for $C\sim\mu$ and uniformly random $x\in\{0,1\}^n$~\footnote{See, e.g.~\cite{hangleiter2018anticoncentration} for a formal proof},
\[
\Pr\!\bigl[p_x(C)\ge \alpha/2^n\bigr]\ \ge\ \beta.
\]
\end{assumption}

\begin{assumption}[Average-case approximate hardness of RQC]\label{fact:mult-hard}
Fix any $\varepsilon_{\mathrm{mult}}(n)=1/\mathrm{poly}(n)$.  
Any probabilistic polynomial-time classical algorithm that, on a $1/\mathrm{poly}(n)$ fraction of pairs
$(C,x)\sim\mu\times\{0,1\}^n$, outputs $\widetilde{p}(x)$ with
\[
(1-\varepsilon_{\mathrm{mult}})\,p_x(C)\ \le\ \widetilde{p}(x)\ \le\ (1+\varepsilon_{\mathrm{mult}})\,p_x(C)
\]
would imply $\#P\subseteq\text BPP$~\footnote{See, e.g.~\cite{bouland2018quantum} for a qualitatively similar statement}.
\end{assumption}

As before, define $F(\psi,\phi)=|\braket{\psi}{\phi}|^2$ and $\mathrm{T}(\psi,\phi)=\sqrt{1-F(\psi,\phi)}$.

\begin{lemma}[Fidelity from many multiplicative amplitudes]\label{lem:fid-to-mult}
Let $\ket{\psi}=C\ket{0^n}$, define a simulator output $\ket{\phi}$, and let $S:=\{x:\ p_\psi(x)\ge \alpha/2^n\}$.
With probability at least $\beta$ (over $C$), we have $|S|\ge \beta\,2^n$, and for at least a $\beta/2$ fraction of $x\in S$,
\[
\frac{|p_\phi(x)-p_\psi(x)|}{p_\psi(x)}\ \le\ \frac{4\,\sqrt{1-F(\psi,\phi)}}{\alpha\beta}\,.
\]
\end{lemma}

\begin{proof}
Data processing gives $\mathrm{TV}(p_\psi,p_\phi)\le \mathrm{T}(\psi,\phi)$ for pure states.
Summing $|p_\phi-p_\psi|$ over $S$ and dividing by the lower bound $\alpha/2^n$ yields the average relative-error bound; applying Assumption~\ref{fact:ac} and Markov’s inequality gives the desired result.
\end{proof}
For peaked circuits with peakedness close to 1, one could prove that an accurate estimation to its peakedness also means an accurate estimation on the state fidelity.

\begin{lemma}[Peak-to-fidelity lower bound]\label{lem:peak-to-fid}
Let the ideal RPC output be
$\ket{\psi_{\mathrm{out}}}=\sqrt{\delta}\ket{0^n}+\sqrt{1-\delta}\ket{\chi}$ with $\delta\in(0,1)$,
and let a simulator output $\ket{\phi_{\mathrm{out}}}$ with $\tilde p_0=|\!\braket{0^n}{\phi_{\mathrm{out}}}\!|^2$
satisfying $|\tilde p_0-\delta|\le \varepsilon_{\mathrm{add}}$ (additive).
Then
\[
F\bigl(\psi_{\mathrm{out}},\phi_{\mathrm{out}}\bigr)\ \ge\ 
F_{\min}(\delta,\varepsilon_{\mathrm{add}})
:=\Bigl(\sqrt{\delta(\delta-\varepsilon_{\mathrm{add}})}-\sqrt{(1-\delta)\bigl(1-\delta+\varepsilon_{\mathrm{add}}\bigr)}\Bigr)^2,
\]
and, for all $\delta\in(0,1)$, $\varepsilon_{\mathrm{add}}\in[0,\delta]$,
\begin{equation}\label{eq:1minusFmin}
1-F_{\min}(\delta,\varepsilon_{\mathrm{add}})\ \le\ 4(1-\delta)\ +\ 2\varepsilon_{\mathrm{add}}.
\end{equation}
\end{lemma}

\begin{proof}
The expression $F_{\min}$ is achieved by choosing the orthogonal parts antiparallel,
$\braket{\chi}{\chi'}=-1$. For \eqref{eq:1minusFmin}, write
$a=\sqrt{\delta(\delta-\varepsilon_{\mathrm{add}})}$, $b=\sqrt{(1-\delta)(1-\delta+\varepsilon_{\mathrm{add}})}$,
use $\sqrt{t(1-x)}\ge t-\frac{x}{2}$ and $\sqrt{u(u+\varepsilon)}\le u+\frac{\varepsilon}{2}$ to get
$a\ge \delta-\varepsilon_{\mathrm{add}}/2$, $b\le (1-\delta)+\varepsilon_{\mathrm{add}}/2$, so
$F_{\min}=(a-b)^2\ge (2\delta-1-\varepsilon_{\mathrm{add}})^2$ and hence
$1-F_{\min}\le 1-(2\delta-1-\varepsilon_{\mathrm{add}})^2\le 4(1-\delta)+2\varepsilon_{\mathrm{add}}$.
\end{proof}

Putting everything together, we show that, if a RPC has peakedness very close to 1, then there is a non-negligible probability that it becomes hard for any classical sequential simulator.

\begin{theorem}[Peakedness barrier for sequential simulators]\label{thm:master}
Fix $\varepsilon_{\mathrm{mult}}(n)=1/\mathrm{poly}(n)$ from Fact~\ref{fact:mult-hard} and define
\[
\tau_\star(n)\ :=\ \Bigl(\frac{\varepsilon_{\mathrm{mult}}\alpha\beta}{4}\Bigr)^2,\qquad
\delta_\star(n)\ :=\ 1-\frac{\tau_\star(n)}{8}.
\]
Let $\varepsilon_{\mathrm{add}}(n)\le \tau_\star(n)/4$.
Then there is no polynomial-time sequential simulator (Def.~\ref{def:simulator}) that, on more than a $1/\mathrm{poly}(n)$ fraction of RPC instances in $\mathcal{F}_n$ with $\delta\ge\delta_\star(n)$, outputs a state $\ket{\phi_{\mathrm{out}}}$ whose peak estimate satisfies
\[
|\tilde p_0-\delta|\ \le\ \varepsilon_{\mathrm{add}}(n).
\]
\end{theorem}

\begin{proof}
Assume for contradiction such a simulator $\mathcal{S}$ exists.
For an instance with $\delta\ge\delta_\star$ and $|\tilde p_0-\delta|\le \varepsilon_{\mathrm{add}}$,
Lemma~\ref{lem:peak-to-fid} and the choice $\delta_\star=1-\tau_\star/8$, $\varepsilon_{\mathrm{add}}\le\tau_\star/4$
give
\[
1-F\bigl(\psi_{\mathrm{out}},\phi_{\mathrm{out}}\bigr)\ \le\ 4(1-\delta)+2\varepsilon_{\mathrm{add}}
\ \le\ 4\cdot \frac{\tau_\star}{8}+2\cdot \frac{\tau_\star}{4}\ =\ \tau_\star,
\]
so $F(\psi_{\mathrm{out}},\phi_{\mathrm{out}})\ge 1-\tau_\star$.
By unitary invariance of fidelity and monotonicity~\eqref{eq:monotone-fid}, the simulator’s
intermediate approximation $\ket{\phi_{\mathrm{int}}}$ (after $C$) satisfies
\[
F\bigl(C\ket{0^n},\ \ket{\phi_{\mathrm{int}}}\bigr)\ \ge\ F\bigl(\psi_{\mathrm{out}},\phi_{\mathrm{out}}\bigr)\ \ge\ 1-\tau_\star.
\]
Applying Lemma~\ref{lem:fid-to-mult} with $\tau=\tau_\star$ shows that, for at least a constant fraction of basis strings $x$ in the anti-concentrated slice (Assumption~\ref{fact:ac}), the resulting output probabilities satisfy
\[
\frac{|p_{\phi_{\mathrm{int}}}(x)-p_x(C)|}{p_x(C)}\ \le\ \frac{4\sqrt{\tau_\star}}{\alpha\beta}
\ =\ \varepsilon_{\mathrm{mult}}.
\]
Thus $\mathcal{S}$ yields a multiplicative $(1\pm\varepsilon_{\mathrm{mult}})$ approximation to $p_x(C)$ on a $1/\mathrm{poly}(n)$ fraction of pairs $(C,x)$, contradicting Assumption~\ref{fact:mult-hard}.
\end{proof}

With stronger assumptions on the circuit output distribution and the simulator model, our bounds can plausibly be strengthened. Motivated by this, we propose the following average-case hardness conjecture:

\begin{conjecture}~\label{conj:bqp_hardness}
    Computing peakedness for a desired string to 1/$\poly(n)$ additive error for random peaked circuits with amplitude $1/\poly(n)$ is average-case BQP-complete.
\end{conjecture}

Below we examine two classical-simulation strategies, either generic or specifically designed for peaked circuits, and explain
why each fails to achieve a 1/$\poly(n)$-additive estimate under reasonable assumptions, providing
circumstantial evidence for Conjecture~\ref{conj:bqp_hardness}.

\paragraph{(i) MPS contraction.}
Even for one–dimensional circuits of depth $T=O(\poly(n))$ the Schmidt rank of
the evolved state grows as $\chi=\exp(\Theta(T))$ after the first
anti-concentration layer~\cite{cirac2021matrix}. Contracting such a network
exactly costs $\chi^{3}$ per slice, i.e.\ $\exp(\poly(n))$ time. Even
approximate MPS truncation fails because $C\ket{0}$ is an intermediate state that is almost maximally entangled and cannot be compressed without losing fidelity extensively.

\paragraph{(ii) Peaked shallow circuit simulation.}
The recent “peaked shallow quantum circuits” algorithm~\cite{bravyi2023classical} by Bravyi, Gosset, and Liu relies heavily on the lightcone argument. They prove that if each output bit depends on only $O(1)$ other bits (which is true for constant-depth circuits), then almost all probability mass in the Pauli basis lies inside a Hamming ball of small radius and thus permit efficient simulation. In the our setting, however, the intermediate state carries weights in exponentially many basis states, so the Pauli list therefore blows up to $2^{\Theta(n)}$, and the algorithm no
longer runs in sub-exponential time.

\medskip
Taken together, these inapproximability results suggest that any
polynomial-time classical routine for $1/\poly(n)$-additive estimation of
$p_{C}(0^n)$ would be either impossible or require a fundamentally new idea, lending credence to Conjecture~\ref{conj:bqp_hardness}.




\subsection{Compiler as an obfuscator}
 Here we give another piece of evidence why deciding the peakedness of $CC'$ is hard, even though their unitary matrices share some elements in common. In particular, we show that, once the unitary is compiled into a quantum circuit, it is very hard to decide whether two elements in these circuits are correlated without extensively considering all other elements in the matrix.

\begin{lemma}
Given a standard compiler that converts a unitary to an exponential-sized circuit, it is impossible to tell whether $C_{ij}$ and $C'_{ij}$ match on a particular value without accessing the compiler exponentially many times.
\end{lemma}
\begin{proof}
This can be shown by considering the standard Givens rotation decomposition for a generic $n$-qubit unitary $C$:

$$
C = G_1 G_2 \cdots G_{N}
$$

where each $G_k$ is a two-level unitary (Givens rotation), and $N \sim O(2^{2n})$ for a generic $2^n \times 2^n$ unitary.

Suppose we want to determine whether the $(i,j)$-th entries of $C$ and $C'$ are equal, i.e., $C_{ij} = C'_{ij}$ or perhaps both take on a particular value. In practice, the matrix element $C_{ij}$ is a complicated function of all the angles in the decomposition:

$$
C_{ij} = f(\{\theta_k\}_{k=1}^{N}),
$$

where each $\theta_k$ is a rotation parameter appearing in the gate sequence output by the compiler. The explicit dependence is highly nonlocal: even a single row or column operation mixes the entire remaining row or column, so that after a few steps, each matrix element is an entangled function of exponentially many parameters. Specifically, in QR or Givens decomposition, the elimination of each entry updates all subsequent rows \& columns recursively, and thus the computation of $C_{ij}$ depends on all previous rotations.

Due to our postselection procedure, the vast majority of the entries of $C_{ij}$ and $C'_{ij}$ will be different. Therefore, to compare $C_{ij}$ and $C'_{ij}$, or to determine if they match a given value, it is not sufficient to locally inspect a small number of gates or entries; in the generic case, you must know all $2^{2n}$ gate parameters or, equivalently, perform an exponential number of queries to reconstruct these entries. Hence, the only generic way to determine if $C_{ij} = C'_{ij}$ is to simulate the full action of both circuits or to explicitly reconstruct the matrices, which is an exponentially hard task for large $n$.
\end{proof}
The bottom line is, even a standard compiler provides a strong obfuscation effect: unless you perform exponentially many queries, it is infeasible to deduce whether two circuits $C$ and $C'$ agree on a particular matrix entry. In our construction, the gates in $C$ and $C'$ look locally Haar-random, with global constraints between the two circuits that cannot be distinguished pair-wisely.

\section{Verifiable quantum advantage with RPC sampling}\label{sec: advantage}
In this section, we turn the “peaked circuit” idea from a hardness statement into a practical recipe that can run on today’s devices. 
We first prove that while the postselection idea is an important theoretical model, it is impractical for problem generation as the success probability for any $\delta\geq1/\poly(n)$ becomes exponentially small even in the $k$-design case. Nevertheless, we show that one could use numerical optimization to search for peaked circuits, and that when random initialization is implemented, these searches yield peaked circuits with properties matching RPCs generated by postselection. Thirdly, we discuss a circuit stitching idea that allows one to scalable construct peaked circuits from small peaked blocks. Lastly, we discuss a practical advantage of peaked circuits and their robustness to sparse bit-flip errors.

\subsection{Bounding the postselection success probability}
\label{sec:postselection}

In the previous sections we built peaked circuits via postselection. However, this literal
postselection is not scalable in $n$ as we explain here: if we draw $C,C'\in U(2^n)$ at random (either truly Haar or from a
$k$-design) then, the probability of first column of $C'$ being
(near-)aligned with that of $C$ is exponentially small, so the acceptance rate of the postselection procedure in Def.~\ref{def: pcc} is exponentially
small in the Hilbert-space dimension $d=2^n$.
Let $d=2^n$ and write $\ket{c}$ for the first column of $C\in U(d)$, i.e.\ $\ket{c}=C\ket{0^n}$, we have the following results:

\begin{lemma}[Peaked Haar random circuits are rare]
\label{lem:haar-cap}
Fix any unit vector $\ket{v}\in\mathbb{C}^d$.  If $\ket{\psi}$ is Haar random on the unit sphere, then
\[
\Pr\!\big[\,|\!\braket{v}{\psi}\!|^2 \ge \delta\,\big]=(1-\delta)^{\,d-1}\,.
\]
Consequently, for $d=2^n$ this probability is doubly exponential in $n$ for any fixed $\delta<1$.
\end{lemma}

\begin{proof}
For complex Haar measure, $X:=|\!\braket{v}{\psi}\!|^2\sim \mathrm{Beta}(1,d-1)$ with density
$(d-1)(1-x)^{d-2}$ on $x\in[0,1]$, hence
$\Pr[X\ge \delta]=\int_{\delta}^1 (d-1)(1-x)^{d-2}\,dx=(1-\delta)^{d-1}$.
\end{proof}

\begin{lemma}[Peaked $k$-design random circuits are rare]
\label{rem:designs}
Exact complex projective $k$-designs match Haar moments up to degree $k$. If $C'$ is drawn from an
exact unitary $k$-design, then its first column $\ket{c'}=C'\ket{0^n}$ is distributed as an
exact state $k$-design; hence for any fixed unit vector $\ket{v}$ and any integer
$t\le k$,
\[
\mathbb{E}\big[|\!\braket{v}{c'}\!|^{2t}\big]
\;=\;
\frac{t!\,(d-1)!}{(d-1+t)!}
\;\le\;
\frac{t!}{d^{\,t}}\,.
\]
By Markov’s inequality, for any threshold $\tau\in(0,1)$,
\[
\Pr\!\big[\,|\!\braket{v}{c'}\!|^2\ge \tau\,\big]
\;\le\;
\frac{\mathbb{E}\big[|\!\braket{v}{c'}\!|^{2t}\big]}{\tau^{\,t}}
\;\le\;
\frac{t!}{(\tau d)^{\,t}}\,.
\]
Setting $\tau=\delta$ and choosing $t=k$ yields the explicit bound
\[
\Pr\!\big[\,|\!\braket{v}{c'}\!|^2\ge \delta\,\big]
\;\le\;
\frac{k!}{\big(\delta\,d\big)^{k}}
\;\le\;
\Big(\frac{c\,k}{\delta\,d}\Big)^{\!k}
\]
for some universal constant $c>0$. Thus, even a $k$-design guarantees a decay like
$d^{-k}=2^{-nk}$ up to poly$(k)$ factors, which is still exponentially small in $n$ for any fixed $k\ge1$.

\end{lemma}

From the above lemma, we see that generating RPCs by fix $C$ and draw $C'$ independently from a $k$-design by postselection requires exponential trials to obtain a constant chance of success—exponential in $n$, hence is not a scalable generator of peaked circuits.

\subsection{Random peaked circuits from variational search}

In practice, we may leverage a variational circuit and variationally search for a $C'$ so that its first column approximates a target $\ket{c}$ up to fidelity $\delta$. 
Under reasonable assumptions, we next show that random seeding the optimizer returns a random representative among many degenerate realizations of $\ket{c}$, effectively obfuscating the underlying circuit
decomposition.
\begin{definition}[Degenerate realizations and equivalence]
\label{def:equiv}
Write $C'\sim \tilde C'$ if $C'$ and $\tilde C'$ induce the same peaked instance $P=C'^{\dagger}C$
up to global phase (equivalently, they have the same first column).  The equivalence class
$\mathcal{O}(\ket{c'})=\{\,\tilde C'\in U(d):\ \tilde C'\ket{0^n}=e^{i\phi}\ket{c'}\,\}$ is a
high-dimensional submanifold of $U(d)$; distinct classes are typically far apart in the
natural Riemannian metric.
\end{definition}

\begin{theorem}[Obfuscation with fixed $C$ via randomized initialization and local optimization]
\label{thm:obfuscation}
Fix $C\in\U(d)$ and let $\ket{c}:=C\ket{0^n}$. Consider losses that depend on $C'$ only through this target column, e.g.
\[
\mathcal{L}(C') \;=\; 1-\big|\!\braket{c}{C'\!0^n}\big|^2
\;=\; 1-\big|\!\bra{0^n}C'^{\dagger}C\ket{0^n}\big|^2,
\]
so minimizing $\mathcal{L}$ is equivalent to maximizing the “peak” of $P:=C'^{\dagger}C$ at $\ket{0^n}$.
Let $\mathcal{A}$ be a local optimizer (e.g., gradient descent, quasi-Newton) that converges to a local minimizer in the basin containing the initialization $\theta_0$.

Define the fiber over the fixed target column
\[
\mathcal{F}_C \;:=\; \Big\{\,C'\in\U(d):\; C'\ket{0^n}=e^{i\phi}\ket{c}\,\Big\}.
\]
Assume:
\begin{enumerate}\itemsep2pt
\item[(i)] (\emph{Fiber minima \& basin separation})
Each connected component of $\mathcal{F}_C$ contains at least one local minimizer of $\mathcal{L}$ achieving $\mathcal{L}=0$, and the basins of attraction of \emph{distinct} such minimizers inside $\mathcal{F}_C$ are disjoint and separated in parameter space by a distance $\Delta>0$.
\item[(ii)] (\emph{Locally uniform seeding})
The initialization distribution $\mu$ over parameters is approximately uniform at scale $\Delta$ (its density is nearly constant on any ball of radius $\Delta$).
\end{enumerate}
Then the output $\mathcal{A}(\theta_0)$ is supported on $\mathcal{F}_C$ and is distributed according to the basin-volume weights under $\mu$. Consequently, the induced peaked circuit
\[
P \;=\; C'^{\dagger}C
\]
has the same first column $\ket{0^n}$ for every run, while its action on the orthogonal $(d-1)$-dimensional subspace varies across runs. In particular, any observable or loss that depends only on the first column of $P$ (e.g., $p_{\max}(P)$) is invariant across runs, whereas the $(d-1)\!\times\!(d-1)$ block of $P$ is randomized by the choice of $C'$.
\end{theorem}

\begin{proof}[Proof sketch]
Because $\mathcal{L}$ depends only on $\braket{c}{C'\!0^n}$, every $C'\in\mathcal{F}_C$ achieves the global minimum $\mathcal{L}=0$. By (i), the parameter space (up to null boundaries) is partitioned into basins $\{B_m\}$ of local minimizers $\{m\}\subset\mathcal{F}_C$. A local method maps each seed $\theta_0$ to the minimizer $m$ whose basin contains it, so the output distribution is the pushforward of $\mu$, assigning weight $\mu(B_m)$ to $m$. For each realized $m$, the resulting peaked circuit is $P=m^{\dagger}C$, which fixes the first column to $\ket{0^n}$ and leaves the orthogonal block determined by $m$; variability of $m$ across runs therefore randomizes that block while leaving first-column observables unchanged.
\end{proof}

Although it is computationally extensive to verify the numerically optimized circuits $C'^\dagger C$ indeed form designs, we provide strong evidence for Thm.~\ref{thm:obfuscation} by showing these circuits satisfy
Obs.~\ref{obs:hs} and shares the same property with those generated by postselection. Specifically, the (unnormalized) Hilbert–Schmidt overlap $\mathbb{E} [|\langle C,C'\rangle_{\mathrm{HS}}|^2]$ is close to $2/d$. To mitigate barren plateau, we consider the following optimization procedure described in Algo.~\cref{algo:1}

\begin{algorithm}[t]
\caption{Numerical search for random peaked circuits}
\begin{algorithmic}[1]
\Require Fix a circuit architecture for both $C$, $C'$.
\State \textbf{Target generation.} Draw each gate of $C$ from the Haar random distribution. Define
\[
P({\Vec{\theta}}):=C'({\Vec{\theta}})^{\dagger}C,\qquad
a({\Vec{\theta}}):=\bra{0^n}P({\Vec{\theta}})\ket{0^n},\qquad
p_0({\Vec{\theta}}):=|a({\Vec{\theta}})|^2.
\]
\State \textbf{Multi-start.} Draw $S$ seeds ${\Vec{\theta}}^{(0,s)}\sim\mathcal{D}_{\rm init}$.
\For{each seed $s$}
  \For{$t=0,\dots,T-1$}
    \State \textit{Evaluate objective.} Build a tensor-network for $p_0({\Vec{\theta}}^{(t,s)})$ and contract.
    \State \textit{Evaluate gradient.} Compute $\nabla_{\Vec{\theta}} p_0({\Vec{\theta}}^{(t,s)})$ by differentiating through the contraction, using
    \[
    \nabla_{\Vec{\theta}} p_0=2\,\mathrm{Re}\!\big(\overline{a}\,\nabla_{\Vec{\theta}} a\big).
    \]
    \State \textit{Adam update.} Minimize $\mathcal{L}({\Vec{\theta}}):=-p_0({\Vec{\theta}})$ with Adam:
    \[
    {\Vec{\theta}}^{(t+1,s)}\leftarrow \mathrm{Adam}\!\big({\Vec{\theta}}^{(t,s)},\,\nabla_{\Vec{\theta}} \mathcal{L}({\Vec{\theta}}^{(t,s)})\big).
    \]
  \EndFor
\EndFor
\State \textbf{Select.} Return the best seed
${\Vec{\theta}}_\star=\arg\max_{s,t} p_0({\Vec{\theta}}^{(t,s)})$,
and set $C'_\star:=C'({\Vec{\theta}}_\star)$, reporting $\delta_\star:=1-p_0({\Vec{\theta}}_\star)$.
\end{algorithmic}\label{algo:1}
\end{algorithm}

We tested this across system sizes $n\in\{8,9,10,11,12\}$, taking both $C$ and the ansatz $C'({\Vec{\theta}})$
to be brickwall circuits of the same depth, $n$. For each $n$ we ran $100$ independent trials (with different target
$C$ and random initialization for $C'$), optimized $C'$ with Adam to reach the prescribed peakedness, and then
computed $\mathbb{E} [|\langle C,C'\rangle_{\mathrm{HS}}|^2]$ by exact tensor–network contraction. 
Figure~\ref{fig:2} shows that the average HS overlap remains close to $2/d$ for all $n$, with no
discernible dependence on system size, in agreement with Obs.~\ref{obs:hs}.

\subsection{Constructing Large Peaked Circuits from circuit stitching}
On a laptop, previous work~\cite{aaronson2024verifiable} successfully identified peaked circuits with sizes up to 36 qubits. With more advanced computational resources like GPUs or TPUs, it is feasible to scale to even larger circuits. Of course, if one could perform classical optimization, then they could spoof by contracting the circuits. However, this issue can be mitigated by ``stitching" peaked circuits together in both horizontal and vertical directions. The observation is, that by combining two random peaked circuits together, the resultant circuit is still peaked. Fig.~\ref{fig:2} left gives an example of such a construction. On the one hand, the verifier would know where the peak is and have a good estimation of the peakedness. On the other hand, since each block is drawn at random and each block is incompressible, it's not hard to show that the whole peaked circuit is non-compressible and thus hard to simulate for a classical challenger.

\begin{lemma}[Exponential decay of peakedness under random block mixing]\label{lem:exp-decay}
Let $d=2^n$ and fix a unit vector $|\psi_0\rangle\in\C^d$. For each layer $j=1,\dots,L$, write
\[
U_j \ =\
\begin{pmatrix}
\alpha_j & a_j^\dagger\\[2pt]
b_j & X_j
\end{pmatrix}
\quad\text{in the decomposition } \C^d=\mathrm{span}\{|\psi_0\rangle\}\oplus |\psi_0\rangle^\perp,
\]
and define the (per-layer) leakage
\[
\varepsilon_j\ :=\ \|b_j\|_2^2\ =\ \|a_j\|_2^2\ =\ 1-|\alpha_j|^2.
\]
Assume that $|\alpha_j|^2$ is very close to 1 and $X_j\in\U(d-1)$ are independent unitary $2$-designs, independent of $(\alpha_j,a_j,b_j)$ and of $\{U_1,\ldots,U_{j-1}\}$. Let
\[
q_j\ :=\ \mathbb{E}\,\bigl[\,|\langle\psi_0|\,U_jU_{j-1}\cdots U_1\,|\psi_0\rangle|^2\,\bigr],
\qquad q_0=1.
\]
Then $q_j$ obeys the one-step recurrence
\begin{equation}\label{eq:one-step-recurrence}
q_j\ =\Bigl(1-\tfrac{d}{d-1}\,\varepsilon_j\Bigr)\,q_{j-1}\ +\ \tfrac{\varepsilon_j}{d-1},
\end{equation}
and hence
\begin{equation}\label{eq:closed-form}
q_L\;-\;\tfrac{1}{d}\ =\
\Biggl(\ \prod_{j=1}^L \Bigl(1-\tfrac{d}{d-1}\,\varepsilon_j\Bigr)\ \Biggr)\ \Bigl(1-\tfrac{1}{d}\Bigr).
\end{equation}
In particular, if $\varepsilon_j\equiv \varepsilon\in(0,1)$ is constant across layers, then
\begin{equation}\label{eq:exp-decay-constant-eps}
q_L\ =\ \tfrac{1}{d}\ +\ \Bigl(1-\tfrac{d}{d-1}\,\varepsilon\Bigr)^{\!L}\,\Bigl(1-\tfrac{1}{d}\Bigr),
\end{equation}
i.e.\ the peakedness decays exponentially in $L$ towards the uniform baseline $1/d$ (rate $\approx e^{-\varepsilon L}$ for large $d$).
\end{lemma}

\begin{proof}
Let $|\psi_{j-1}\rangle:=U_{j-1}\cdots U_1|\psi_0\rangle$ and decompose it as
$|\psi_{j-1}\rangle=\gamma_{j-1}\,|\psi_0\rangle + |w_{j-1}\rangle$ with $|w_{j-1}\rangle\in|\psi_0\rangle^\perp$ and $|\gamma_{j-1}|^2=q_{j-1}$. Then
\[
\langle\psi_0|U_j|\psi_{j-1}\rangle
=\alpha_j\,\gamma_{j-1}\ +\ a_j^\dagger X_j\,|w_{j-1}\rangle.
\]
Conditioning on $U_1,\ldots,U_{j-1}$ and on $(\alpha_j,a_j)$, the unitary $2$-design property of $X_j$ implies
$\mathbb{E}_{X_j}[a_j^\dagger X_j|w_{j-1}\rangle]=0$ and
$\mathbb{E}_{X_j}\!\bigl[\,|a_j^\dagger X_j|w_{j-1}\rangle|^2\,\bigr]=\|a_j\|_2^2\,\|w_{j-1}\|_2^2/(d-1)
=\varepsilon_j(1-q_{j-1})/(d-1)$.
Taking expectations gives
\[
q_j=\mathbb{E}\bigl[|\alpha_j|^2\bigr]\,q_{j-1}+\frac{\mathbb{E}[\varepsilon_j]}{d-1}\,(1-q_{j-1}).
\]
Since $|\alpha_j|^2=1-\varepsilon_j$ deterministically by unitarity of $U_j$, this reduces to
\[
q_j=(1-\varepsilon_j)q_{j-1}+\frac{\varepsilon_j}{d-1}(1-q_{j-1})
=\Bigl(1-\tfrac{d}{d-1}\varepsilon_j\Bigr)q_{j-1}+\frac{\varepsilon_j}{d-1},
\]
which is \eqref{eq:one-step-recurrence}. Solving the affine recursion yields \eqref{eq:closed-form}; the specialization \eqref{eq:exp-decay-constant-eps} follows by taking $\varepsilon_j\equiv\varepsilon$.
\end{proof}

In numerical optimization, as long as we choose $\varepsilon$ sufficiently small: $\varepsilon\sim1/L$, the expectation of the composed circuit will remain a $O(1)$ number.

 Can someone classically spoof this by contracting each small peaked circuit $U_1$, $U_2$..., and spot where the peak is? The trick is, that the classical challenger wouldn't know where the segments are, and they cannot infer the peak from randomly selecting a subregion of the circuit. Asymptotically, breaking a large circuit into $k=\Theta(\log n)$ contiguous peaked blocks already yields a super-polynomial number of stitching patterns: if the circuit has $m$ gates (or layers) in a fixed topological order, the number of ways to place $k-1$ cuts is $\binom{m-1}{k-1}$; for $m=\Theta(n)$ this is $2^{\,\Omega((\log n)^2)}$, and in $2$D with $m=\Theta(n^2)$ it is even larger. A natural attempt to detect the true pattern is to look for “high concentration” at block boundaries via single– or few–qubit marginals, but this is unreliable:  local circuit rewrite rules can shift or disperse any apparent concentration across a few neighboring gates while keeping the overall circuit equivalent. Under the rewrite, the output distribution will not go through high concentration at the boundary. In practice, local obfuscations and rewrite protocols also can destroy an explicit circuit pattern, which increases the combinatorial ambiguity. 

\subsection{Robustness to sparse classical noise}\label{sec:robustness-formal}

One merit of RPCs is their robustness to many realistic noise models. Local errors tend to shrink structure and gently fill the distribution toward uniform, but the heavy outcome remains detectable as long as its excess over uniform stays above statistical error. Here we model residual measurement errors as classical bit flips and analyze how one could recover peakedness (as well as the peaked string) under two ubiquitous classical noise models:

\begin{itemize}
\item[(1)] \textbf{$t$-sparse bit-flip (pre-measurement \& readout) noise}:
in each shot, an adversary (or a stochastic process) flips at most $t$ output bits
right before computational-basis measurement.\vspace{2pt}

\item[(2)] \textbf{I.I.D.\ readout bit-flips}:
each measured bit is flipped independently with probability $r\in[0,1/2)$
(a binary-symmetric channel, BSC$(r)$).
\end{itemize}

\subsubsection{Estimating $p_{\max}$ with Hamming-ball aggregation.}
We first assume $x_\star$ is known to the verifier and ask whether we could give a good estimation of $p_{\max}$. The intuition is that small bit flips will result a small Hamming ball that is centered around $x_\star$.
For $t\in\{0,1,\ldots,n\}$ write the Hamming ball
$B_t(x_\star):=\{x\in\{0,1\}^n:\ \mathrm{dist}(x,x_\star)\le t\}$ and its size
$|B_t|:=|B_t(x_\star)|=\sum_{h=0}^t\binom{n}{h}$. 
aggregate the observed mass in a small Hamming ball:

\begin{algorithm}[H]
\caption{Hamming-Ball Aggregation (HBA$(t)$)}
\begin{algorithmic}[1]
\Require $N$ samples $x_1,\ldots,x_N \in \{0,1\}^n$ drawn from $P$, reference $x_\star\in\{0,1\}^n$, radius $t\in\{0,\ldots,n\}$
\Ensure $\widehat{p}^{(t)} \approx \Pr_{x\sim P}\!\big[\mathrm{dist}(x,x_\star)\le t\big]$
\State $c \gets 0$
\For{$i = 1$ to $N$}
  \If{$\mathrm{dist}(x_i,x_\star) \le t$}
    \State $c \gets c + 1$
  \EndIf
\EndFor
\State \Return $\widehat{p}^{(t)} \gets c / N$
\end{algorithmic}
\end{algorithm}
\paragraph{(A) $t$-sparse per-shot flips (adversarial).}

\begin{proposition}[Adversarial $t$-sparse bit-flips]
\label{prop:t-sparse}
With radius parameter $t$, for \emph{every} outcome distribution $p(\cdot)$ of $P$
and every realization of the flips,
\[
p_{\max} \;\le\; \mathbb{E}[\widehat{p}^{(t)}]\;=\; \sum_{x\in B_t(x_\star)} p(x)
\;\le\;p_{\max} \;+\; b\,|B_t|,
\qquad\text{where}\quad b:=\max_{x\neq x_\star} p(x).
\]
In particular, $\widehat{p}^{(t)}$ \emph{lower-bounds} $p_{\max}$, and its upward bias is at most $b\,|B_t|$.
\end{proposition}

\begin{proof}
A $t$-flip can change any bitstring by Hamming distance at most $t$. Hence, if the ideal shot equals
$x_\star$ then the observed shot still lies in $B_t(x_\star)$, contributing to $\widehat{p}^{(t)}$;
consequently, $p_{\max} \le \mathbb{E}[\widehat{p}^{(t)}]$.
Conversely, only shots whose ideal strings already lie in $B_t(x_\star)$ can be mapped into $B_t(x_\star)$
by $\le t$ flips, so $\mathbb{E}[\widehat{p}^{(t)}]=\sum_{x\in B_t(x_\star)} p(x) \le p_{\max}+b(|B_t|-1)$.
\end{proof}
\paragraph{(B) I.I.D.\ bit-flip noise.}
\begin{remark}[Bias is negligible for peaked ensembles]
\label{rmk:b-small}
For our peaked construction the non-peak weights look $k$-design-like “random”
at low order (Haar baseline $\approx 2^{-n}$),
so typically $b=O(2^{-n})$ while $p_{\max} \gg 2^{-n}$.%

Thus the bias bound $b|B_t|$ is $|B_t|/2^n\ll 1$ even for $t=O(\log n)$, because
$|B_t| \le \sum_{h\le t}\binom{n}{h} \le (en/t)^t = \exp\!\big(O(\log^2 n)\big)\ll 2^n$.
\end{remark}

\begin{theorem}[I.I.D.\ readout flips admit a logarithmic error budget]
\label{thm:bsc-logn}
Under \emph{(N2)} with rate $r$, let $W\sim\mathrm{Bin}(n,r)$ be the number of flipped bits in one shot.
Choose any $t\ge (1+\delta)\,nr$ with fixed $\delta\in(0,1]$.
Then
\[
\mathbb{E}[\widehat{p}^{(t)}]
\;\ge\;p_{\max} \cdot \Pr[W\le t]
\;\ge\;p_{\max} \cdot \Big(1-e^{-\frac{\delta^2}{2+\delta}\,nr}\Big),
\]
and the upward bias is at most $b\,|B_t|$ as in Prop.~\ref{prop:t-sparse}.
In particular, if $nr=\Theta(\log n)$ and $t=\Theta(\log n)$ then with probability $1-n^{-\Omega(1)}$
every shot from $x_\star$ falls inside $B_t(x_\star)$ (Chernoff bound), while contamination is
$O(|B_t|/2^n)=o(1)$ for peaked ensembles. Hence HBA$(t)$ preserves the peak signal
with at most a vanishing additive error.
\end{theorem}

\begin{proof}
Condition on the ideal outcome. If the ideal outcome is $x_\star$ (probability $p_{\max}$),
the observed shot lies in $B_t(x_\star)$ whenever $W\le t$. This yields the stated lower bound.
All other ideal outcomes contribute at most $b\,|B_t|$ by the same argument as
Prop.~\ref{prop:t-sparse}. The Chernoff bound $\Pr[W>(1+\delta)nr]\le \exp(-\frac{\delta^2}{2+\delta}nr)$
is standard.
\end{proof}

\subsubsection{Recovering the peaked string when $x_\star$ is unknown.}
Next we discuss how to recover the peaked string with noisy measurement outcomes when the string is unknown.
We observe i.i.d.\ measurement outcomes $X^{(1)},\dots,X^{(N)}\in\{0,1\}^n$ drawn from the (noisy)
output distribution of $C$. Being consistent with last subsection, we consider the same two noise models
(A) Adversarial
$t$-sparse flips (adversarial but bounded) (B) i.i.d.\ bit-flip noise $BSC(r)$ with $r<\tfrac12$; . We show how to recover $x_\star$ and then estimate the
peak robustly by the same Hamming-ball estimator used when $x_\star$ is known.

\begin{definition}[Hamming distance and ball]
For $x,y\in\{0,1\}^n$, let $\mathrm{dist}_H(x,y)$ be the Hamming distance, and let
$B_{t}(x):=\{y:\mathrm{dist}_H(x,y)\le t\}$. Write $|B_t|=\sum_{h=0}^{t}\binom{n}{h}$.
\end{definition}

\paragraph{(A) $t$-sparse per-shot flips (adversarial).}
Assume each shot differs from the clean outcome by at most $t$ bit flips (the error pattern may be
adversarial and vary across shots). Consider the \emph{Hamming-center} decoder:

\begin{algorithm}[H]
\caption{Hamming-Center$(t)$}
\label{alg:hamming-center}
\begin{algorithmic}[1]
\Require Samples $X^{(1)},\dots,X^{(N)}\in\{0,1\}^n$, radius $2t$.
\State For each $i$, compute the cluster size
$c_i:=\big|\{j:\mathrm{dist}_H(X^{(i)},X^{(j)})\le 2t\}\big|$.
\State Let $i^\star\in\arg\max_i c_i$ and define the core
$\mathcal{C}:=\{j:\mathrm{dist}_H(X^{(j)},X^{(i^\star)})\le 2t\}$.
\State Output $\hat s$ as the bitwise majority over $\{X^{(j)}:j\in\mathcal{C}\}$.
\end{algorithmic}
\end{algorithm}

\begin{proposition}[Recovery under $t$-sparse flips]
\label{prop:sparse-recovery}
Suppose a fraction $p_{\max}$ of the shots are $t$-flipped versions of $x_\star$, and the remaining
shots are arbitrary. Then there exist constants $c_1,c_2>0$ such that if
\[
N\;\ge\; c_1\,\frac{|B_{2t}|}{p_{\max}^{2}}\;\log\!\frac{n}{\eta},
\]
Algorithm~\ref{alg:hamming-center} returns $\hat s=x_\star$ with probability at least $1-\eta$.
\emph{Sketch.} Any two $t$-flipped copies of $x_\star$ are within distance $\le 2t$, so the true
cluster contributes $\approx p_{\max} N$ points to a single $2t$-ball; spurious points do not
concentrate in any such ball. A counting/Chernoff argument shows the densest ball is dominated by
noisy copies of $x_\star$, and the bitwise majority within that ball yields $x_\star$.
\end{proposition}

As mentioned above, once $\hat s$ is recovered, one could estimate the peak by $\hat p(t)$ around $\hat s$; with
$t'=O(t)$ the estimator concentrates at $p_{\max}$ with standard Hoeffding-type rates.

\paragraph{(B) I.I.D.\ bit-flip noise.}
Under $BSC(r)$ each bit is flipped independently with probability $r<\frac12$.
Let $\bar X_j:=\frac1N\sum_{i=1}^N X^{(i)}_j$ be the empirical mean of bit $j$.
Define the \emph{bitwise majority decoder}
\[
\hat s_j\;:=\;\mathbb{1}\!\left\{\bar X_j\ge\tfrac12\right\},\qquad j=1,\dots,n.
\]

\begin{proposition}[Majority recovers $x_\star$ under $BSC(r)$]
\label{prop:majority}
If the peak has weight $p_{\max}>0$ and the background mass has no systematic per-bit bias,\footnote{Formally,
$\big|\Pr[X_j=1\mid X\neq x_\star]-\tfrac12\big|\le o(1)$, which holds for near-uniform residual mass.} then each bit has bias toward $x_\star^j$ of at least
$\frac12\,p_{\max}(1-2r)$, and from a Chernoff bound for any failure probability $\eta\in(0,1)$,
\[
N\;\ge\; \frac{c\,\log(n/\eta)}{p_{\max}^{2}(1-2r)^{2}}
\quad\Longrightarrow\quad
\Pr\big[\hat s=x_\star\big]\;\ge\;1-\eta,
\]
for a universal constant $c>0$. 
\end{proposition}

Given $\hat s$, estimate the peak weight by the same Hamming-ball statistic as above:
\[
\hat p(t)\;:=\;\frac1N\sum_{i=1}^{N}\mathbb{1}\!\left\{\mathrm{dist}_H\!\big(X^{(i)},\hat s\big)\le t\right\},
\]
with $t$ chosen per the noise level (e.g., $t\approx r n$ or a small multiple thereof)
\subsection{Recovering peakedness under weak global depolarizing noise}
Lastly, we examine the robustness of peaked circuits to a quantum noise channel: the global depolarizing channel. Let the true peak be $p_{\max}=p(x_\star)$ as usual. We assume in execution, the quantum circuit goes through of the global depolarizing channel with strength $\varepsilon$~\footnote{Random circuits turns local noise into white noise~\cite{dalzell2021random}. Therefore if we characterize the gate level noises sufficiently well we can assume we have a good estimation on the global depolarizing rate},
\[
p'_{\max}=(1-\varepsilon)\,p_{\max}+\frac{\varepsilon}{2^{n}}.
\]
From $N$ samples with $X\sim\mathrm{Bin}\!\left(N,\,p'_{\max}\right)$ and
$\widehat{p}'_{\max}=X/N$, define the de-biased estimator
\[
\widehat{p}_{\max}\;=\;\frac{\widehat{p}'_{\max}-\varepsilon/2^{n}}{1-\varepsilon}.
\]
Then $\mathbb{E}[\widehat{p}_{\max}]=p_{\max}$ and
\[
\operatorname{SE}(\widehat{p}_{\max})
=\frac{\sqrt{p'_{\max}(1-p'_{\max})/N}}{1-\varepsilon}
\;\le\;\frac{1}{2(1-\varepsilon)\sqrt{N}}.
\]
To achieve additive error $|\widehat{p}_{\max}-p_{\max}|\le \alpha$ with failure
probability $\le \delta$, it suffices (by Chernoff) to take
\[
N=\Theta\!\left(\frac{p'_{\max}(1-p'_{\max})}{(1-\varepsilon)^2\,\alpha^2}
\log\frac{1}{\delta}\right)
=\Theta\!\left(\frac{1}{(1-\varepsilon)^2\,\alpha^2\,p_{\max}}
\log\frac{1}{\delta}\right),
\]
using $p'_{\max}\asymp p_{\max}$. For a peaked circuit with
$p_{\max}=O(1)$, this becomes
\[
N=\Theta\!\left(\frac{1}{(1-\varepsilon)^2\,\alpha^2}
\log\frac{1}{\delta}\right).
\]

Therefore, as long as the survival probability $1-\varepsilon = \Omega(1/\poly(n)$ one could recover the peakedness to some polynomial additive error with $N = \poly(n)$ samples.

\section{Future Directions}

Our work examines the properties of random peaked circuits: on one specific input string, these circuits carry a single designated computational-basis string $x_\star$ carries an anomalously large weight $p_{\max}(P)=|\!\braket{x_\star}{P|0^n}\!|^2\ge \delta$, while on all other inputs the output distribution look almost Haar random. We show strong analytical evidence that this task of estimating peakedness is hard for a classical simulator, while its quantum simulation remain near-term implementable and verifiable. In this sense, RPCs offer a minimal-structure requirement to quantum-classical separation: peaked quantum circuits do not require careful design (like Shor's algorithm) to be classically hard. On the other hand, it showed that, unlike in the RQC sampling case, anti-concentration is not a necessary condition for a quantum circuit to be average case hard. These random peaked circuits can be used as the next generation of sampling-based quantum advantage protocol for efficient classical verification.

While our work focuses on resolving existing open problems for RPCs, a few future goals suggest themselves:

\paragraph{Numerical and experimental demonstration at large scale.}
(i) \emph{Numerical scaling:} First of all, it would be interesting push the exact-contraction/Adam synthesis of Sec.~\ref{sec: advantage} to larger $n$ and depth, and verify if the analytical predictions still hold for the optimized circuits at that scale. (ii) \emph{Hardware demonstrations:} Another immediate next step is to implement stitched RPCs (introduced in Sec.~\ref{sec: advantage}) at system and circuit sizes beyond those used for random-circuit sampling, leveraging the product bound $|\!\braket{x_L}{U|x_0}\!|^2\ge\prod_i(1-\delta_i)$. Further, it would be interesting to rigorously prove relate the hardness of simulating such stitched and locally obfuscated RPCs to known complexity classes.


\paragraph{Peaked circuits as an encryption protocal.}
Thus far we have assumed (w.l.o.g.) that the input is $\ket{0^n}$. As a theoretical perspective, we
advocate a \emph{secret–hiding} variant in which the input is unknown and the task is to decide whether the circuit is peaked:
\[
\text{Given a circuit } P,\ \text{decide whether }\exists\,x, y\in\{0,1\}^n\ \text{s.t. }\
|\!\braket{x}{P|y}\!|^2 \,\ge\, \delta .
\]
This is reminiscent of the \emph{Non-Identity Check} problem~\cite{janzing2005non,ji2009non} which is known to be QMA-complete, but with a complementary
flavor: there, one asks whether a circuit is close to identity on all inputs or far on
some witness state; here, we instead ask whether there exist a input state where the unitary mapping is nearly trivial (up to some bit-flips).

If this variant is hard even for quantum algorithms (e.g., QMA-/QCMA-hard under natural
promises), then peaked circuits suggest an encryption application: with a key (the hidden state) one
can efficiently decode the hidden peak(s), while without it the instance is computationally
intractable. Beyond a single string, one can hide a set $S\subseteq\{0,1\}^n$ with the number of peaked strings $|S|=\mathrm{poly}(n)$ and total mass $\sum_{x\in S} p(x)$, moving toward
code-like “hidden sets’’ in Hilbert space that are efficiently decodable and
plausibly hard to recover otherwise.

\paragraph{Obfuscation by compilation.}
In Sec.~\ref{sec: hardness} we have seen that a standard, deterministic compiler can be used to hide information. However, even when two circuits implement the same unitary, checking their equivalence from gate lists alone can be hard. Consider the following thought experiment: we draw a random polynomial-depth circuit $C_1$. Suppose we have the computational power to obtain its matrix $U$, and then we resynthesize $U$ via a standard, e.g., cosine–sine decomposition to obtain $C_2$. Clearly $C_1$ and $C_2$ represent the same unitary. 

But given only $C_1$ and $C_2$, deciding whether they are equivalent seems extremely costly: a naive test requires simulating the two circuits and checking the output closeness to the identity, which is generally exponentially costly. Moreover, small local edits in $C_2$ (e.g., inserting a gate at some random location) can dramatically alter the overall unitary. This suggests hiding randomness in the compilation process to make equivalence checking even harder. For the problem generator, the circuit is effectively an identity transformation of the logical algorithm; for any challenger, $C_1$ and $C_2$ can look very different, and the best available approach is to compute or characterize the full unitary—something a classical computer is likely to fail at for large $n$.

\section{Acknowledgment}

The author thanks Scott Aaronson, Dima Abanin, Bill Fefferman, Hrant Gharibyan, Soumik Ghosh, Hong-Ye Hu, Hayk Tepanyan, Yifan Zhang, and Leo Zhou for their insightful discussions and feedback. YZ was supported by the Natural Science and Engineering Research Council (NSERC) of Canada and acknowledges support from the Center for Quantum Materials and Centre for Quantum Information and Quantum Control at the University of Toronto. Resources used in preparing this research were provided, in part, by the Province of Ontario, the Government of Canada through CIFAR, and companies sponsoring the Vector Institute \url{www.vectorinstitute.ai/#partners}.

In preparation of the manuscript, we are aware of another related work, where the authors consider constructing peaked circuits with a quantum error correction approach called ``Hidden Code Sampling''~\cite{deshpande2025peaked}, which provides another promising path to scaling up peaked circuits that are provably hard.

\appendix
\bibliographystyle{alpha}
\bibliography{main}

\section{An alternative proof of Thm~\ref{thm: incom}}\label{sec: alter}
Here we provide an alternative proof sketch. Fundamentally, this construction is possible because of collisions: two or more distinct unitaries can lead to the same final state. In fact, in a $k$-design, the number of states with complexity $\sim kn$ is $d^k/k!$ while the number of unitaries with the same complexity is roughly $d^{2k}/k!$~\cite{brandao2021models}, where $d := 2^n$. Crucially, the concentration of measure says, in k-designs, the probability $p_i$ of a state (or circuit) being picked should be roughly flat: 
\begin{itemize}
    \item Pick a random state from a k-design (assuming $k>3$) according to $p_i$, with $>1/6$ probability, the corresponding $\Omega(p_i = k!/(d)^{-k})$ (see appendix for a proof);
    \item For all elements from the unitary k-design, max $p_i \lessapprox k!/d^{2k}$ (see, e.g. Lemma 3 of \cite{brandao2021models})
\end{itemize}
This redundancy is the key to this `obfuscation' process.
Then we can consider the following states: $\ket{\psi_1} :=C\ket{0^{n-1}}\otimes\ket{1}$ and $\ket{\psi_1'} :=C'\ket{0^{n-1}}\otimes\ket{1}$. As we show below, by random matrix theory and design properties, with overwhelmingly high probability, they should be far apart. Assume $C'^\dagger C$ is compressible - this violates the assumption because one could easily connect those two states by applying a shortcut circuit $C' C^\dagger$ to $\ket{\psi_1}$, which we show below is forbidden.
\begin{fact}
    With overwhelmingly high probability, $C\ket{0^{n-1}\otimes1}$ and $C'\ket{0^{n-1}\otimes1}$ are far apart.
\end{fact}
How to show this? Well, let us start with the Haar-random case as a motivating example. 
To generate the $i$-th vector in a Haar-random unitary, one could do the following: 
\begin{enumerate}
\item Generate $G$ according to
\[
G \;\in\;\mathbb{C}^{d\times 2},
\qquad
G_{ij}\;\overset{\text{i.i.d.}}{\sim}\;\mathcal{N}(0,1)\;+\;i\,\mathcal{N}(0,1).
\]

  \item Let \(g_{0}\) be the first column of \(G\).  Normalize:
    \[
      \ket{\psi_{0}}
      \;=\;
      \frac{g_{0}}{\|g_{0}\|}\,.
    \]
  \item Let \(g_{1}\) be the second column.  Project orthogonal to \(\ket{\psi_{0}}\) and normalize:
    \[
      \tilde g_{1}
      \;=\;
      g_{1}
      \;-\;
      \ket{\psi_{0}}\bigl\langle \psi_{0}\bigr|g_{1}\bigr),
      \quad
      \ket{\psi_{1}}
      \;=\;
      \frac{\tilde g_{1}}{\|\tilde g_{1}\|}\,.
    \]
\end{enumerate}

Then the pair \(\bigl(\ket{\psi_{0}},\ket{\psi_{1}}\bigr)\) is exactly distributed as the first two columns of a Haar‐random unitary on \(\mathbb{C}^{d}\).

For our construction, the vectors of $C$ are completely selected at random. There, by our demand, $C'$ shares the same first vector as $C$, and its second vector is generated randomly. Geometrically, we are randomly sampling two points on a $2^{n-1}$ dimensional hyperbolic sphere that is orthogonal to $\ket{\psi_0}$. Intuitively, the chance that they are close to each other should be doubly exponentially small.

In fact, $\ket{\psi_1}$ and $\ket{\psi_1'}$ are very close to Haar random states:
\begin{proposition}\label{prop: haar}
        If $C$, $C'$ are picked and postselected from a Haar random ensemble according to Definition~\cref{def: pcc}, then \(\|\psi_{1}-v_1\|=O(d^{-1/2})\) and \(\|\psi_{1}'-v_1'\|=O(d^{-1/2})\)  with $1-\exp(-n)$ high probability.
\end{proposition}
\begin{proof}
Let 
\[
v_1 \;=\;\frac{g_{1}}{\|g_{1}\|}\;\sim\;\mathrm{Haar}\bigl(S^{2d-1}\bigr),
\qquad
f(v_1)\;=\;\bigl|\langle \psi_{0}\mid v_1\rangle\bigr|.
\]
Since \(f\) is 1–Lipschitz on \(S^{2d-1}\), Lévy’s lemma implies that there is a universal constant \(c>0\) such that for all \(\varepsilon>0\),
\[
\Pr\bigl(f(v_1)\ge\varepsilon\bigr)
\;\le\;
2\,\exp\!\bigl(-c\,(2d-1)\,\varepsilon^{2}\bigr).
\]
Moreover, from the Gram–Schmidt construction one shows
\[
\bigl\|\ket{\psi_{1}} - v_1\bigr\|
\;\le\;
2\,f(v_1).
\]
Combining these,
\[
\Pr\bigl(\|\psi_{1}-v_1\|\ge 2\varepsilon\bigr)
\;\le\;
2\,\exp\!\bigl(-c\,(2d-1)\,\varepsilon^{2}\bigr),
\]
so \(\|\psi_{1}-v_1\|=O(d^{-1/2})\) with $1-\exp(-d)$ high probability.
But, even the circuit complexity between two $\varepsilon$-approximate  $n$-qubit Haar random states requires exponential gates to implement. Therefore, there cannot exist a short-cut between these two states.
\end{proof}

For a $k$-design in the $(d-1)$-dimensional subspace, one recovers Haar-like behavior up to $k$-th moments. Using the same convention \(\|\ket{\psi_{1}} - v\|\le 2\,f(v)\) and applying Markov’s inequality to the \(k\)th moment for an \(k\)-design, one obtains
\[
\Pr\bigl(\|\psi_{1}-v\|\ge 2\varepsilon\bigr)
\;\le\;
\frac{\mathbb{E}[\,f(v)^{k}\bigr]}{\varepsilon^{k}}
\;=\;
O\!\Bigl(\frac{1}{(N^{1/2}\,\varepsilon)^{k}}\Bigr)\;.
\]
In this case, Propasition~\ref{prop: haar} becomes:
\begin{proposition}\label{prop: k-design}
        If $C$, $C'$ are picked and postselected from a $k$-design according to Definition~\ref{def: pcc}, then \(\|\psi_{1}-v_1\|=O(d^{-0.49})\) and \(\|\psi_{1}'-v_1'\|=O(d^{-0.49})\)  with $1-\exp(-n)$ high probability.
\end{proposition}

A $k$-design `only' guarantees $O(d^{-k/2}$) tail decay, but even just setting $\varepsilon = O(d^{-0.49})$ still gives an exponentially small in $n$. Crucially, this shows that, even if we fix the first column vector of $C'$ to be the same as $C$, with almost probability 1, $\ket{\psi_1}'$ still forms $\varepsilon$-approximate $k$ designs themselves with $\varepsilon$ exponentially small in $n$.

We are now ready to prove the main theorem: we want to show that with high probability that connecting $\ket{\psi_1}'$ into $\ket{\psi_1}$ requires many gates. Intuitively, this is justified, as two states drawn from an approximate $k$-design are very likely to be far apart in the circuit complexity picture. Our next step is to formalize this by showing a counting argument.
Let $\mathcal{D}=\{\ket{\phi_1},\dots,\ket{\phi_N}\}\subset \mathbb{C}^d$ be an $\varepsilon$–approximate spherical $k$–design (uniform measure). 
Fix a reference state $\ket{\phi'}$ (which could be any state as the design properties are independent of the reference state) and $0<\delta<1$. 
Define the Fubini–Study distance $d_{\mathrm{FS}}(\ket{\phi},\ket{\phi'})=\sqrt{1-|\langle\phi|\phi'\rangle|^2}$.

\begin{lemma}
The number $P_\delta$ of pairwise $\delta$–separated states contained in $\mathcal{D}$ satisfies
\[
P_\delta \;\ge\; \frac{(1-\delta^2)^k}{1+\varepsilon}\binom{d+k-1}{k}
\;=\; \frac{(1-\delta^2)^k}{(1+\varepsilon)k!}\,d^{\,k}\Bigl(1+O\!\bigl(\tfrac{k^2}{d}\bigr)\Bigr).
\]
In particular, for fixed $k,\delta,\varepsilon$, $P_\delta = \Theta(d^{k})$ and hence 
\[
N \;\ge\; P_\delta \;=\; \Omega(d^{k}).
\]
\end{lemma}

\begin{proof}
Let $f(\phi)=|\langle\phi|\phi'\rangle|^{2k}$. For Haar measure
$\mathbb{E}_{\text{Haar}} f = \binom{d+k-1}{k}^{-1}$. Since $\mathcal D$ is an
$\varepsilon$–approximate $k$–design (uniform over its $N$ points),
\[
\frac1N\sum_{j=1}^N f(\phi_j)\;\le\;(1+\varepsilon)\binom{d+k-1}{k}^{-1}.
\]
Write $B_\delta(\phi)=\{\phi: d_{\mathrm{FS}}(\phi,\phi')\le\delta\}$; on this event
$|\langle\phi|\phi\rangle|^{2}\ge 1-\delta^{2}$, hence
$f(\phi)\ge (1-\delta^{2})^{k}$. Let $q:=|\mathcal D\cap B_\delta(\phi)|/N$.
Then
\[
(1-\delta^{2})^{k} q \;\le\; \frac1N\sum_{j=1}^N f(\phi_j)
\;\le\; (1+\varepsilon)\binom{d+k-1}{k}^{-1},
\]
so
\[
q \;\le\; \frac{1+\varepsilon}{(1-\delta^{2})^{k}} \binom{d+k-1}{k}^{-1}.
\]
Thus any Fubini–Study ball of radius $\delta$ contains at most
$qN \le (1+\varepsilon)(1-\delta^{2})^{-k}\binom{d+k-1}{k}^{-1} N$ points.
A greedy packing that successively selects a point and removes all points
within distance $\delta$ produces at least
\[
P_\delta \;\ge\; \frac{N}{qN} \;\ge\;
\frac{(1-\delta^{2})^{k}}{1+\varepsilon}\binom{d+k-1}{k}.
\]
This matches the stated bound. The asymptotic form follows from
$\binom{d+k-1}{k} = d^{k}/k!(1+O(k^{2}/d))$. 
\end{proof}
The large number of distinct states in an approximate design allows us to prove the following circuit complexity fact:

\begin{lemma}[Circuit size lower bound from packing]
\label{lem:circuit-lb}
Fix $\delta\in(0,1)$ and a finite universal two–qubit gate set $\mathcal G$.
Let $\mathcal S$ be an ensemble of $N$ pure $n$–qubit states (with $d=2^n$)
that are pairwise Fubini–Study distance $\ge \delta$. Any circuit over $\mathcal G$
using at most $s$ two–qubit gates (in a fixed layout, e.g.\ brickwork) can generate
at most
\[
  \Bigl(\frac{C s}{\delta}\Bigr)^{\Gamma s}
\]
distinct states from $\mathcal S$, where $c,\Gamma>0$ depend only on $\mathcal G$.
Consequently
\begin{equation}\label{eq:implicit-s}
  s \;\ge\; \frac{\log N}{\Gamma\bigl(\log s + \log(c/\delta)\bigr)}.
\end{equation}
In particular, for constant $\delta$ and large $N$,
\[
  s \;=\; \Omega\!\left(\frac{\log N}{\log\log N}\right).
\]
If moreover $N=\Theta(d^{k})=\Theta(2^{k n})$ for constant $k$, then
\[
  s \;=\; \Omega\!\left(\frac{k n}{\log(k n)}\right).
\]
\end{lemma}

\begin{proof}

For any precision $\alpha\in(0,1)$ there is an $\alpha$–net
$\mathcal G_\alpha\subset SU(4)$ with
$|\mathcal G_\alpha|\le (c/\alpha)^{\Gamma}$ (Solovay--Kitaev plus compilation
of single–qubit gates; constants absorbed into $c,\Gamma$).

Let $U=U_s\cdots U_1$ and let $\tilde U_j$ approximate $U_j$ with
$\|U_j-\tilde U_j\|_\infty\le \alpha$.
Telescoping gives
\[
  \|U-\tilde U\|_\infty
  \;\le\; \sum_{j=1}^s
  \bigl\| U_s\cdots U_{j+1}(U_j-\tilde U_j)\tilde U_{j-1}\cdots\tilde U_1\bigr\|_\infty
  \;\le\; s\alpha.
\]
Thus $\|U\ket{\phi} - \tilde U\ket{\phi}\|_2 \le s\alpha$ for any input $\ket{\phi}$. Set $\alpha := \delta/(3s)$. Then every length‑$s$ circuit has a discretized
representative producing an output within $\delta/3$ of the original. Fix a length‑$s$ two–qubit gate layout. The number of discretized circuits of
length $\le s$ is at most
\[
  M(s,\delta) \;\le\; |\mathcal G_\alpha|^{s}
  \;\le\; \Bigl(\frac{c}{\alpha}\Bigr)^{\Gamma s}
  \;=\; \Bigl(\frac{c s}{\delta}\Bigr)^{\Gamma s}.
\]
Each discretized circuit corresponds to a ball (radius $\delta/3$) covering all
outputs of circuits within per–gate $\alpha$.

Two distinct $\delta/3$–balls cannot contain two states of $\mathcal S$ at
distance $\ge \delta$ (triangle inequality). Hence $N \le M(s,\delta)$, giving
\[
  \log N \;\le\; \Gamma s\bigl(\log s + \log(c/\delta)\bigr),
\]
which rearranges to \eqref{eq:implicit-s}.

For constant $\delta$, $\log(c/\delta)=O(1)$, so
$s\log s = \Omega(\log N)$ and therefore
$s = \Omega(\log N / \log\log N)$. Substituting $N=\Theta(2^{k n})$ yields
$s = \Omega(k n / \log(k n))$ for constant $k$.
\end{proof}
It is worth noticing that the lower bound is a global geometric/combinatorial fact about how many well‑separated states short circuits can reach. Changing the ``origin'' for any fixed starting state, either $\ket{0}$ or $\ket{\psi_1}$ in our problem does not change the fact. One could prove the following corollary:

\begin{corollary}[Reference--independence of typical circuit lower bound]
\label{cor:reference-independence}
Let $\mathcal D$ be an $\varepsilon$--approximate spherical $k$--design on $n$ qubits
whose elements are pairwise Fubini--Study distance at least $\delta\in(0,1)$, and let
$N = |\mathcal D| = \Theta(2^{k n})$ (for fixed $k,\delta,\varepsilon$).
Fix an \emph{arbitrary} reference state $\ket{\eta}$.
For any $s\in\mathbb N$ denote by $\mathcal C(\ket{\eta}\!\to\!\ket{\psi})$ the minimal number
of two--qubit gates from a fixed finite universal gate set needed to map $\ket{\eta}$
to $\ket{\psi}$ up to Euclidean error $\le \delta/3$.
Then there exist constants $c,\Gamma>0$ (depending only on the gate set) such that
\begin{equation}\label{eq:ref-indep-tail}
  \Pr_{\ket{\psi}\sim \mathcal D}\Big[\,\mathcal C(\ket{\eta}\!\to\!\ket{\psi}) < s\,\Big]
  \;\le\;
  \frac{(c s/\delta)^{\Gamma s}}{N}.
\end{equation}
Choosing $s = c_1 k n / \log(k n)$ with $c_1>0$ sufficiently small yields
\[
  \Pr_{\ket{\psi}\sim \mathcal D}\Big[\,\mathcal C(\ket{\eta}\!\to\!\ket{\psi})
     < c_1 k n / \log(k n)\Big]
  \;\le\; e^{-\Omega(k n)}.
\]
Thus, with probability $1 - e^{-\Omega(k n)}$ over a random $\ket{\psi}\in\mathcal D$,
\[
  \mathcal C(\ket{\eta}\!\to\!\ket{\psi})
  = \Omega\!\Bigl(\frac{k n}{\log(k n)}\Bigr),
\]
uniformly for every fixed choice of $\ket{\eta}$.
\end{corollary}

\begin{proof}
Let $W$ be any unitary with $W\ket{0^{\otimes n}} = \ket{\eta}$ and define the rotated
ensemble $\mathcal D' := W^\dagger \mathcal D$.
Unitary invariance of the Fubini--Study metric implies $\mathcal D'$ is also
$\delta$--separated, and since (approximate) spherical $k$--designs are invariant
under global conjugation, $\mathcal D'$ remains an $\varepsilon$--approximate
$k$--design of the same size $N$.

For any $\ket{\psi}\in\mathcal D$ we have
$\mathcal C(\ket{\eta}\!\to\!\ket{\psi}) = \mathcal C(\ket{0^{\otimes n}}\!\to\! W^\dagger\ket{\psi})$,
hence
\[
  \Pr_{\ket{\psi}\sim\mathcal D}\!\big[\mathcal C(\ket{\eta}\!\to\!\ket{\psi}) < s\big]
  =
  \Pr_{\ket{\psi'}\sim\mathcal D'}\!\big[\mathcal C(\ket{0^{\otimes n}}\!\to\!\ket{\psi'}) < s\big].
\]
The right-hand probability is bounded by the absolute (reference $\ket{0}$) case:
by the circuit counting / packing lemma (Lemma~\ref{lem:circuit-lb}) there are at most
$(c s/\delta)^{\Gamma s}$ states of circuit size $< s$ in any $\delta$--separated set,
establishing \eqref{eq:ref-indep-tail}. Taking $s = c_1 k n/\log(k n)$ and choosing
$c_1$ small so that
$\Gamma s \log(c s/\delta) \le \tfrac{1}{2} k n$ for large $n$ gives the exponential
tail $e^{-\Omega(k n)}$. This yields the stated high-probability lower bound.
\end{proof}

Combining Corollary~\ref{cor:reference-independence} and Proposition~\ref{prop: k-design} gives Theorem~\ref{thm: incom}: for a fixed choice of $C$, with high probability, the second column vector of $C'$ forms approximate designs and should be far from the second column vector of $C$ with high probability.

\section{Tail Bound from a State 3-Design}

Let $\{\ket{\psi_i}\}_{i=1}^N \subset \mathbb{C}^d$ be a spherical $3$-design, here $N=O(d^3)$ is the number of states in the design (and in general, $
N
= \binom{d + k - 1}{k}
\;\approx\;
\frac{d^k}{k!}.
$ for $k$-designs), so
\[
\frac{1}{N}\sum_{i=1}^N \ket{\psi_i}\!\bra{\psi_i} = \frac{\mathbb{I}}{d},
\qquad
\frac{1}{N}\sum_{i=1}^N (\ket{\psi_i}\!\bra{\psi_i})^{\otimes m}
= \int_{\text{Haar}} (\ket{\psi}\!\bra{\psi})^{\otimes m}\, d\psi
\quad \text{for } m=1,2,3.
\]
Fix any reference state $\ket{\phi}\in\mathbb{C}^d$, and define
\[
  p_i := \frac{d}{N}\,|\langle\phi|\psi_i\rangle|^2,
  \qquad I \sim p,\qquad X := p_I.
\]
Since $\sum_i |\langle\phi|\psi_i\rangle|^2 = N/d$ by 1-design property, $\{p_i\}$ is a valid probability distribution.
Our goal is to lower bound $\Pr\big[X \ge 1/N\big]$, where $X$ is the random variable that we sample according to $p$.

Because the set is a $3$-design, we can replace the first three even moments of $|\langle\phi|\psi\rangle|$ by their Haar values:
\[
\mathbb{E}_{\text{Haar}}[|\langle\phi|\psi\rangle|^{2m}]
= \frac{m!\,(d-1)!}{(d+m-1)!},\qquad m=1,2,3.
\]
Thus
\begin{align*}
\mathbb{E}[X] 
&= \sum_{i=1}^N p_i^2
 = \Big(\tfrac{d}{N}\Big)^2 \sum_{i=1}^N |\langle\phi|\psi_i\rangle|^4
 = \frac{d^2}{N^2}\cdot N \cdot \frac{2}{d(d+1)}
 = \frac{2d}{N(d+1)},
\\[4pt]
\mathbb{E}[X^2]
&= \sum_{i=1}^N p_i^3
 = \Big(\tfrac{d}{N}\Big)^3 \sum_{i=1}^N |\langle\phi|\psi_i\rangle|^6
 = \frac{d^3}{N^3}\cdot N \cdot \frac{6}{d(d+1)(d+2)}
 = \frac{6 d^{2}}{N^{2}(d+1)(d+2)}.
\end{align*}

Now we recall Paley--Zygmund inequality: for a nonnegative r.v.\ $X$ and $\theta\in(0,1)$,
\[
\Pr\!\big[X \ge \theta\,\mathbb{E}[X]\big]
\;\ge\;
(1-\theta)^2\,\frac{\mathbb{E}[X]^2}{\mathbb{E}[X^2]}.
\]
We set the threshold $t := 1/N$ and choose
\[
\theta = \frac{t}{\mathbb{E}[X]}
= \frac{1/N}{2d/[N(d+1)]}
= \frac{d+1}{2d},
\]
so $1-\theta = \frac{d-1}{2d}$. Plugging in,
\[
\Pr\!\Big[X \ge \frac{1}{N}\Big]
\ge
\left(\frac{d-1}{2d}\right)^{\!2}
\frac{\left(\frac{2d}{N(d+1)}\right)^2}{\frac{6d^2}{N^2(d+1)(d+2)}}
=
\frac{(d-1)^2(d+2)}{6\,d^2(d+1)}.
\]

As such we reach the conclusion:
\begin{theorem}
Let $\{\ket{\psi_i}\}_{i=1}^N$ be an spherical $3$-design in $\mathbb{C}^d$ ($d\ge2$). 
Fix any $\ket{\phi}$ and define
$p_i=\tfrac{d}{N}|\langle\phi|\psi_i\rangle|^2$, $I\sim p$, and $X=p_I$.
Then
\[
\Pr\!\Big[X \ge \frac{1}{N}\Big]
\;\ge\;
\frac{(d-1)^2(d+2)}{6\,d^2(d+1)}.
\]
This lower bound is independent of $N$ and approaches $1/6$ as $d\to\infty$.
\end{theorem}
\section{A upper bound on the pair-wise gate correlation}\label{appx: gate}
In the main text Sec.\ref{sec: complexity}, under $1$-peakedness assumption we proved
\[
P:=C'{}^{\dagger}C \ = \,\diag(1,V),
\qquad V\sim \text{Haar on }\U(d-1),\ \ d=2^n,
\]
with $R$ any basis alignment sending $\ket{0^n}$ to the tracked ray. Equivalently, conditioning on
$P\ket{0^n}=\ket{0^n}$, the restriction of $P$ to $\ket{0^n}^\perp$ is Haar. In particular,
\[
\mathbb{E}\big[\,|\Tr P|^2\,\big] \ = \ 1+\mathbb{E}\big[\,|\Tr V|^2\,\big] \ = \ 2,
\qquad
\mathbb{E}\big[\|P-\mathbb{I}\|_F^2\big]
\ = \ \mathbb{E}\big[\|V-\mathbb{I}_{d-1}\|_F^2\big] \ = \ 2(d-1).
\]

Here we want to show that this condition also restricts the pair-wise correlation between gates in $C$ and $C'$. Suppose that, fixing $C$, every same-position gate in $C'$ has high overlap with its counterpart in $C$,
\[
\rho_m:=\frac{\big|\Tr\!\big((C'_m)^\dagger C_m\big)\big|^2}{D_m}\ \ge\ 1-\varepsilon
\quad(\text{all }m),
\]
then the telescoping bound yields
\[
\|P-\mathbb{I}\|_F \ \le\ \sum_m \|(C'_m)^\dagger C_m-\mathbb{I}\|_F
\ \le\ M\,\sqrt{d\,\varepsilon},
\]
so $\|P-\mathbb{I}\|_F^2\le M^2 d\,\varepsilon$. For any ensemble supported on such $C'$, this forces
\[
\mathbb{E}\big[\|P-\mathbb{I}\|_F^2\big]\ \le\ M^2 d\,\varepsilon,
\]
which contradicts the conditional-Haar expectation
$\mathbb{E}\big[\|P-\mathbb{I}\|_F^2\big]=2(d-1)$ unless $M^2\varepsilon=\Omega(1)$.
Thus, an ensemble built by keeping all gates of $C'$ highly correlated with those of $C$ cannot be
close (in distribution) to the conditional-Haar law on $\ket{0^n}^\perp$; in particular, it cannot
approximate a unitary $1$-design there, let alone a $k$-design.

\end{document}